\documentclass[10pt,journal,twocolumn]{IEEEtran}

\usepackage{amsfonts,amsbsy,bm,amsmath}
\usepackage[nolist]{acronym}
\usepackage{mathrsfs} 
\usepackage{subcaption}
\usepackage{multicol}

\usepackage{graphicx,cite,amssymb,mathtools,amsthm}
\usepackage{stfloats}
\usepackage{xcolor}
\usepackage{bbm}
\usepackage{pgfplots}
\usepackage{tikz}
\usetikzlibrary{arrows,shapes.misc,chains,scopes}
\usetikzlibrary{spy}
\pgfplotsset{compat=1.15}
\usetikzlibrary{decorations.pathreplacing}
\usepackage{siunitx}
\usepackage{csquotes}
\usepackage{makecell}
\usepackage[ruled,linesnumbered,vlined]{algorithm2e}
\usepackage{stmaryrd}

\usetikzlibrary{matrix}
\usepgfplotslibrary{groupplots,units}
\pgfplotsset{compat=newest}

\SetNlSty{bfseries}{\color{black}}{}
\SetArgSty{textnormal}

\DeclareMathOperator*{\argmin}{argmin}

\newcommand{\nonl}{\renewcommand{\nl}{\let\nl\oldnl}}

\usepackage{nicefrac}
\newcommand{\mcc}{\textcolor{black}}
\newcommand{\ed}{\textcolor{black}}

\newcommand{\edd}{\textcolor{black}}
\newcommand{\eddd}{\textcolor{black}}
\newcommand{\edddd}{\textcolor{black}}
\newcommand{\eddddd}{\textcolor{black}}

\begin{document}
	\newcommand{\ensemble}{\mathscr{C}}
\newcommand{\code}{\mathcal{C}}
\newcommand{\vecu}{\boldsymbol{u}}
\newcommand{\veci}{\boldsymbol{i}}
\newcommand{\vecf}{\boldsymbol{f}}
\newcommand{\vecv}{\boldsymbol{v}}
\newcommand{\vecp}{\boldsymbol{p}}
\newcommand{\vecx}{\boldsymbol{x}}
\newcommand{\vecone}{\boldsymbol{1}}
\newcommand{\vecuhat}{\hat{\boldsymbol{u}}}
\newcommand{\vecc}{\boldsymbol{c}}
\newcommand{\vecb}{\boldsymbol{b}}
\newcommand{\vecchat}{\hat{\boldsymbol{c}}}
\newcommand{\vecy}{\boldsymbol{y}}
\newcommand{\vecz}{\boldsymbol{z}}
\newcommand{\B}{\boldsymbol{B}}
\newcommand{\G}{\boldsymbol{G}}
\newcommand{\GK}{\boldsymbol{K}}
\newcommand{\Gsys}{\G_{\mathsf{i},\mathsf{sys}}}
\newcommand{\Gnsys}{\G_{\mathsf{i},\mathsf{nsys}}}
\newcommand{\Per}{\boldsymbol{\Pi}}
\newcommand{\I}{\boldsymbol{I}}
\newcommand{\perP}{\boldsymbol{P}}
\newcommand{\perS}{\boldsymbol{S}}
\newcommand{\GH}[1]{\bm{\mathsf{G}}_{#1}}
\newcommand{\T}[1]{\bm{\mathsf{T}}{(#1)}}

\newcommand{\mi}{\mathrm{I}}
\newcommand{\Prob}{P}
\newcommand{\Z}{\mathrm{Z}}
\newcommand{\SPC}{\mathcal{S}}
\newcommand{\Rep}{\mathcal{R}}
\newcommand{\f}{\mathrm{f}}
\newcommand{\SC}{\mathrm{SC}}
\newcommand{\Q}{\mathrm{Q}}
\newcommand{\MAP}{\mathrm{MAP}}
\newcommand{\E}{\mathrm{E}}
\newcommand{\U}{\mathrm{U}}
\newcommand{\Ham}{\mathrm{H}}

\newcommand{\Amin}{\mathrm{A}_{\mathrm{min}}}
\newcommand{\dmin}{d}
\newcommand{\erfc}{\mathrm{erfc}}

\newcommand{\de}{\mathrm{d}}

\newcommand{\decoRule}{\rule{\textwidth}{.4pt}}

\newcommand{\oleq}[1]{\overset{\text{(#1)}}{\leq}}
\newcommand{\oeq}[1]{\overset{\text{(#1)}}{=}}
\newcommand{\ogeq}[1]{\overset{\text{(#1)}}{\geq}}
\newcommand{\ogeql}[2]{\overset{#1}{\underset{#2}{\gtreqless}}}
\newcommand\numeq[1]%
{\stackrel{\scriptscriptstyle(\mkern-1.5mu#1\mkern-1.5mu)}{=}}

\newcommand{\pscd}{\gl{\Prob(\mathcal{E}_{\SCD})}}
\newcommand{\uED}{\gl{\hat{u}}}
\newcommand{\LED}{\gl{L_i^{(\ED)}}}
\newcommand{\inverson}[1]{\gl{\mathbb{I}\left\{#1\right\}}}
\theoremstyle{definition}
\newtheorem{mydef}{Definition}
\newtheorem{prop}{Proposition}
\newtheorem{theorem}{Theorem}

\newtheorem{lemma}{Lemma}
\newtheorem{remark}{Remark}
\newtheorem{example}{Example}
\newtheorem{definition}{Definition}
\newtheorem{corollary}{Corollary}


\definecolor{lightblue}{rgb}{0,.5,1}
\definecolor{normemph}{rgb}{0,.2,0.6}
\definecolor{supremph}{rgb}{0.6,.2,0.1}
\definecolor{lightpurple}{rgb}{.6,.4,1}
\definecolor{gold}{rgb}{.6,.5,0}
\definecolor{orange}{rgb}{1,0.4,0}
\definecolor{hotpink}{rgb}{1,0,0.5}
\definecolor{newcolor2}{rgb}{.5,.3,.5}
\definecolor{newcolor}{rgb}{0,.3,1}
\definecolor{newcolor3}{rgb}{1,0,.35}
\definecolor{darkgreen1}{rgb}{0, .35, 0}
\definecolor{darkgreen}{rgb}{0, .6, 0}
\definecolor{darkred}{rgb}{.75,0,0}
\definecolor{midgray}{rgb}{.8,0.8,0.8}
\definecolor{darkblue}{rgb}{0,.25,0.6}

\definecolor{lightred}{rgb}{1,0.9,0.9}
\definecolor{lightblue}{rgb}{0.9,0,0.0}
\definecolor{lightpurple}{rgb}{.6,.4,1}
\definecolor{gold}{rgb}{.6,.5,0}
\definecolor{orange}{rgb}{1,0.4,0}
\definecolor{hotpink}{rgb}{1,0,0.5}
\definecolor{darkgreen}{rgb}{0, .6, 0}
\definecolor{darkred}{rgb}{.75,0,0}
\definecolor{darkblue}{rgb}{0,0,0.6}

\definecolor{bgblue}{RGB}{245,243,253}
\definecolor{ttblue}{RGB}{91,194,224}

\definecolor{dark_red}{RGB}{150,0,0}
\definecolor{dark_green}{RGB}{0,150,0}
\definecolor{dark_blue}{RGB}{0,0,150}
\definecolor{dark_pink}{RGB}{80,120,90}

\begin{acronym}
\acro{5G}{the $5$-th generation wireless system}
        \acro{ANV}{average number of node-visits}
	\acro{APP}{a-posteriori probability}
	\acro{ARQ}{automated repeat request}
	\acro{ASCL}{adaptive successive cancellation list}
	\acro{ASK}{amplitude-shift keying}
	\acro{AUB}{approximated union bound}
	\acro{AWGN}{additive white Gaussian noise}
	\acro{B-DMC}{binary-input discrete memoryless channel}
	\acro{BEC}{binary erasure channel}
	\acro{BER}{bit error rate}
	\acro{biAWGN}{binary-input additive white Gaussian noise}
	\acro{bpcu}{bits per channel use}
	\acro{BPSK}{binary phase-shift keying}
	\acro{BRGC}{binary reflected Gray code}
	\acro{BSS}{binary symmetric source}
	\acro{CC}{chase combining}
	\acro{CN}{check node}
	\acro{CRC}{cyclic redundancy check}
	\acro{CSI}{channel state information}
	\acro{DE}{density evolution}
	\acro{DMC}{discrete memoryless channel}
	\acro{DMS}{discrete memoryless source}
        \acro{dRM}{dynamic RM}
	\acro{DSCF}{dynamic successive cancellation flip}
	\acro{eMBB}{enhanced mobile broadband}
	\acro{FER}{frame error rate}
	\acro{uFER}{undetected frame error rate}
	\acro{FHT}{fast Hadamard transform}
	\acro{GA}{Gaussian approximation}
	\acro{GF}{Galois field}
	\acro{HARQ}{hybrid automated repeat request}
	\acro{i.i.d.}{independent and identically distributed}
	\acro{IF}{incremental freezing}
	\acro{IR}{incremental redundancy}
	\acro{LDPC}{low-density parity-check}
	\acro{LFPE}{length-flexible polar extension}
	\acro{LHS}{left hand side}
	\acro{LLR}{log-likelihood ratio}
	\acro{MAP}{maximum-a-posteriori}
	\acro{MC}{Monte Carlo}
	\acro{MLC}{multilevel coding}
	\acro{MLPC}{multilevel polar coding}
	\acro{MLPC}{multilevel polar coding}
	\acro{ML}{maximum-likelihood}
	\acro{MC}{metaconverse}
        \acro{ORBGRAND}{ordered reliability bits guessing random additive noise decoding}
	\acro{PAC}{polarization-adjusted convolutional}
	\acro{PAT}{pilot-assisted transmission}
	\acro{PCM}{polar-coded modulation}
	\acro{PDF}{probability density function}
	\acro{PE}{polar extension}
	\acro{PMF}{probability mass function}
	\acro{PM}{path metric}
	\acro{PW}{polarization weight}
	\acro{QAM}{quadrature amplitude modulation}
	\acro{QPSK}{quadrature phase-shift keying}
	\acro{QUP}{quasi-uniform puncturing}
	\acro{RCU}{random-coding union}
	\acro{RHS}{right hand side}
	\acro{RM}{Reed-Muller}
	\acro{RQUP}{reversal quasi-uniform puncturing}
	\acro{RV}{random variable}
	\acro{SC-Fano}{successive cancellation Fano}
	\acro{SCOS}{successive cancellation ordered search}
	\acro{SCF}{successive cancellation flip}
	\acro{SCL}{successive cancellation list}
	\acro{SCS}{successive cancellation stack}
	\acro{SC}{successive cancellation}
	\acro{SE}{spectral efficiency}
	\acro{SNR}{signal-to-noise ratio}
	\acro{SP}{set partitioning}
    \acro{TEP}{test error pattern}
	\acro{UB}{union bound}
	\acro{VN}{variable node}
\end{acronym}
	\title{Successive Cancellation Ordered Search\\ Decoding of Modified $\G_N$-Coset Codes}

	\author{\IEEEauthorblockN{Peihong Yuan, \IEEEmembership{Member, IEEE}, and Mustafa Cemil Co\c{s}kun}, \IEEEmembership{Member, IEEE}
	    \thanks{This work was partially supported by the German Research Foundation (DFG) under Grant KR~3517/9-1. This paper was presented in part at the IEEE Information Theory Workshop (ITW), October 2021, Kanazawa, Japan~\cite{YC21} and some of the results are published in Chapter 4 of the Ph.D. thesis of the first author~\cite{yuan2021thesis}.}
	\thanks{Peihong Yuan is with the Research Laboratory of Electronics (RLE), Massachusetts Institute of Technology (MIT), Cambridge, MA, USA (email: phyuan@mit.edu). Mustafa Cemil Co\c{s}kun is with Radio Systems Research Lab of Nokia Bell Labs, Murray Hill, NJ, USA (mustafa.coskun@nokia-bell-labs.com). Parts of this work has been carried out when both authors were with the Institute for Communications Engineering (LNT), Technical University of Munich (TUM), Munich, Germany.}
}

	\maketitle

\begin{abstract}
A tree search algorithm called successive cancellation ordered search (SCOS) is proposed for $\G_N$-coset codes that implements maximum-likelihood (ML) decoding with adaptive complexity for transmission over binary-input AWGN channels. Unlike bit-flip decoders, no outer code is needed to terminate decoding; therefore, SCOS also applies to $\G_N$-coset codes modified with dynamic frozen bits. The average complexity is close to that of successive cancellation (SC) decoding at practical frame error rates (FERs) for codes with wide ranges of rate and lengths up to $512$ bits, which perform within $0.25$ dB or less from the random coding union bound and outperform Reed--Muller codes under ML decoding by up to $0.5$ dB. Simulations illustrate simultaneous gains for SCOS over SC-Fano, SC stack (SCS) and SC list (SCL) decoding in FER and the average complexity at various SNR regimes. SCOS is further extended by forcing it to look for candidates satisfying a threshold, thereby outperforming basic SCOS under complexity constraints. The modified SCOS enables strong error-detection capability without the need for an outer code. In particular, the $(128, 64)$ polarization-adjusted convolutional code under modified SCOS provides gains in overall and undetected FER compared to CRC-aided polar codes under SCL/dynamic SC flip decoding at high SNR.
\end{abstract}
	\begin{keywords}
            Complexity-adaptive maximum-likelihood decoding, error detection, polar codes, Reed--Muller codes, dynamic frozen bits.
	\end{keywords}
 \markboth
	{to appear in IEEE Transactions on Communications}
	{}

\section{Introduction}\label{sec:intro}

$\G_N$-coset codes are a class of block codes \cite{arikan2009channel} that include polar codes\cite{stolte2002rekursive,arikan2009channel} and \ac{RM} codes\cite{reed,muller}. Polar codes achieve capacity over \acp{B-DMC} under low-complexity \ac{SC} decoding\cite{arikan2009channel} and \ac{RM} codes achieve capacity \mcc{over} \acp{BEC} under \ac{ML} decoding\cite{kudekar17}.\mcc{\footnote{\mcc{\ac{RM} codes achieve capacity over \acp{B-DMC} under bit-wise \ac{ML} decoding~\cite{RP21}, i.e., the average bit error probability vanishes asymptotically in the block length.}}}
However, their performance under \ac{SC} decoding~\cite{stolte2002rekursive} is not competitive for short- to moderate-lengths, e.g., from $\edd{64}$ to \ed{$512$} bits\cite{Coskun18:Survey}.\footnote{Long \ac{RM} codes are not well-suited for \ac{SC} decoding\cite{IU21,CP21}: the error probability of long \ac{RM} codes under \ac{SC} decoding is lower-bounded by $\nicefrac{1}{2}$\cite[Section X]{arikan2009channel}.} Significant research effort hast been put into approaching \ac{ML} performance by \edd{improved decoding algorithms with an} SC decoding schedule\edd{\cite{Dumer:List06,NC12:seq,miloslavskaya2014sequential,trifonov2018score,jeong2019sc,afisiadis2014low,chandesris2018dynamic}, improving the distance properties~\cite{RMpolar14,Mondelli14,trifonov16,KKK19,LGZ21,YFV21,RBV20,MV20} or both\cite{tal2015list,arikan2019sequential}}.

The idea of \emph{dynamic} frozen bits lets one represent any linear block code as a \emph{modified} $\G_N$-coset code~\cite{trifonov16}. This concept unifies the concatenated polar code approach, e.g., with a high-rate outer \ac{CRC} code, to improve the distance spectrum of polar codes so that they can be decoded with low to moderate complexity~\cite{tal2015list}.

This paper proposes \emph{\ac{SCOS}} as \edd{an} \ac{ML} decoder for modified $\G_N$-coset codes. The decoding complexity adapts to the channel noise and an extension of \ac{SCOS} limits the worse-case complexity while still permitting near-\ac{ML} decoding \ed{for various code lengths $N\in\ed{\{64, 128, 256, 512\}}$ and \edd{wide ranges of rate from low to high}}. The decoder can be used for $\G_N$-coset codes\edd{,} \ac{CRC}-concatenated $\G_N$-coset codes \edd{as well as those with dynamic frozen bits}. Numerical results show that \ac{RM} and RM-polar codes with dynamic frozen bits of block length $N\in\ed{\{64, 128, 256, 512\}}$, \edd{i.e., \ac{dRM}~\cite{CP21} and} \ed{the proposed} dRM-polar codes, perform within $0.25$ dB of the \ac{RCU} bound\cite{Polyanskiy10:BOUNDS} with an average complexity close to that of \ac{SC} decoding at \ed{low enough \acp{FER}}. \edd{Remarkably, \ac{dRM} codes under SCOS outperform \ac{ML} performance of \ac{RM} codes up to $0.5$ dB}. \ed{To illustrate the benefits \edd{of the proposed algorithm}, simulations with \edd{\ac{SC-Fano}, \ac{SCS}} and \ac{SCL} decoding \edd{algorithms} are also provided \edd{for codes of length $128$ as examples}. Compared to \edd{all}, SCOS provides lower average decoding complexity with better \ac{FER} performance \edd{at various operating regimes}. \ac{SCOS} is further extended by limiting attention to codeword candidates satisfying an optimized threshold test, which} improves the performance when a maximum complexity constraint is imposed. In addition, the threshold test lets the decoder avoid making a decision\cite{forney1968exponential} which provides simultaneous gains in overall \ac{FER} and \ac{uFER} for the $(128,64)$ \ac{PAC} code as compared to a \ac{CRC}-concatenated polar code under \ac{SCL} \ed{and \ac{DSCF}} decoding \ed{algorithms}, where the \ac{CRC} is optimized for the lower tail of the distance spectrum\cite{Yuan19}.

\ac{SCOS} borrows ideas from \ac{SC}-based flip\cite{afisiadis2014low,chandesris2018dynamic}, sequential\cite{fano1963heuristic,NC12,miloslavskaya2014sequential,trifonov2018score,jeong2019sc} and list decoders\cite{D74,fossorier,wu2007soft,tal2015list}. It is a tree search algorithm that flips the bits of valid paths to find a leaf with higher likelihood than other leaves, if such a leaf exists, and repeats until the \ac{ML} decision is found. \edddd{The search stores the branches} \eddddd{in the list} that is updated progressively while running partial \ac{SC} decoding by flipping the bits of the most-likely leaf at each iteration. The order of the candidates is chosen according to the probability that they provide the \ac{ML} decision. \ac{SCOS} does not require an outer code (as for flip-decoders) or parameter optimization for the performance vs. complexity trade-off (as for sequential decoders).

This paper is organized as follows. Section \ref{sec:prelim} gives background on the problem. Section \ref{sec:sc_os} presents the \ac{SCOS} algorithm \mcc{with the pseudo codes}. \edd{The complexity of \ac{SCOS} is discussed together with the numerical results} in Section \ref{sec:comp}. \edd{Then, Section \ref{sec:improvements} proposes modifications and provide\edddd{s} numerical results in comparison to the original algorithm.} Section \ref{sec:numerical} \edd{compares \ac{SCOS} to the other existing complexity-adaptive decoders} and Section \ref{sec:conclusions} concludes the paper.
\section{Preliminaries}
\label{sec:prelim}
Let $x^a$ be the vector $(x_1, x_2, \dots, x_a)$; if $a=0$, then the vector is empty. Given $x^N$ and a set $\mathcal{A}\subset [N]\triangleq\{1,\dots,N\}$, let $x_{\mathcal{A}}$ be the subvector $(x_i:i\in\mathcal{A})$. For \eddd{set $\mathcal{A}$, we define} an intersection set as $\mathcal{A}^{(i)}\triangleq\mathcal{A}\cap[i]$, $i\in[N]$. Uppercase letters refer to \acp{RV} and lowercase letters to realizations. A \ac{B-DMC} is denoted as $ W : \mathcal{X} \rightarrow \mathcal{Y} $, with input alphabet $ \mathcal{X} = \{0,1\} $, output alphabet $ \mathcal{Y} $, and transition probabilities $ W(y|x) $ for $ x \in \mathcal{X} $ and $ y \in \mathcal{Y} $\edd{\cite[Sec. 4]{Richardson:2008:MCT:1795974}}. The transition probabilities of $N$ independent uses of the same channel are denoted as $ W^N(y^N|x^N)$ and can be factored as  $W^N(y^N|x^N)= \prod_{i=1}^{N} W(y_i|x_i)$. Capital bold letters refer to matrices, e.g., $\B_N$ denotes the $N\times N$ \emph{bit reversal} matrix \cite{arikan2009channel} and $\G_{2}$ denotes the $2\times 2$ Hadamard matrix.

\subsection{$\G_N$-coset Codes}\label{sec:polar_RM}
Consider the matrix $\G_{N} = \B_N\G_{2}^{\otimes n}$, where $N = 2^n$ with a non-negative integer $n$ and $\G_{2}^{\otimes n}$ is the $n$-fold Kronecker product of $\G_{2}$. For the set $\mathcal{A}\subseteq[N]$ with $|\mathcal{A}|=K$, let $U_{\mathcal{A}}$ have entries that are \ac{i.i.d.} uniform \emph{information} bits, and let $U_{\mathcal{A}^c}=u_{\mathcal{A}^c}$ be fixed or \emph{frozen}\ed{, where $\mathcal{A}^c\triangleq [N]\setminus\mathcal{A}$}. The mapping $c^N = u^N\G_{N}$ defines a $\G_N$-coset code\cite{arikan2009channel}. Polar and \ac{RM} codes are $\G_N$-coset codes with different selections of $\mathcal{A}$\cite{arikan2009channel,stolte2002rekursive}.

Using $\G_{N}$, the transition probability from $u^N$ to $y^N$ is $W_{N}(y^N|u^N) \triangleq W^{N}(y^N|u^N\G_{N})$. The transition probabilities of the $i$-th \emph{bit-channel}, a synthesized channel with the input $u_i$ and the output $(y^N,u^{i-1})$, are defined by
\begin{equation}
	W_{N}^{(i)}(y^N,u^{i-1}|u_i) \triangleq \sum_{u_{i+1}^N\in\mathcal{X}^{N-i}}\frac{1}{2^{N-1}}W_{N}(y^N|u^N).
\end{equation}
An $(N,K)$ polar code is designed by placing the $K$ most reliable bit-channels with indices $i\in[N]$ into the set $\mathcal{A}$ that can, e.g., be found using density evolution \cite{arikan2009channel,Mori:2009SIT}. An $r$-th order \ac{RM} code of length-$N$ and dimension $K = \sum_{i=0}^{r}\binom{n}{i}$, where $0\leq r\leq n$, is denoted as RM$(r,n)$. Its set $\mathcal{A}$ consists of the indices, $i\in[N]$, with Hamming weight at least $n-r$ for the binary expansion of $i-1$.  For both codes, one sets $u_i=0$ for $i\in\mathcal{A}^c$.

We make use of \emph{dynamic} frozen bits\cite{trifonov16}. A frozen bit is dynamic if its value depends on a subset of information bits preceding it; the resulting codes are called modified $\G_N$-coset codes. Dynamic frozen bits tend to improve the performance of near-\ac{ML} decoders\cite{Yuan19,arikan2019sequential,CNP20,CP21,LGZ21} because the weight spectrum of the resulting code tends to improve as compared to the underlying code\cite{Yuan19,LZG19,YFV21,RBV20,LZL21,LGZ21}. For the numerical results, we will consider short- to moderate-length \ac{RM} codes with dynamic frozen bits, called dRM codes\cite[Def. 1]{CP21}. An important instance from the ensemble is the PAC codes with \ac{RM} rate-profiling\cite{arikan2019sequential}. However, the average complexity gets large for (near-)\ac{ML} decoding of dRM codes as they get longer. Therefore, \edd{the formal definition of the modified RM-polar codes is given below}.
\begin{definition}\label{def:dRM}
The $(N,K)$ dRM-polar ensemble is the set of codes specified by the set $\mathcal{A}$ of an $(N,K)$ RM-polar code and choosing
\begin{equation}
    u_i = \underset{j\in\mathcal{A}^{(i-1)}}{\bigoplus} v_{j,i}u_j, \quad\forall i\in\mathcal{A}^c\label{eq:dfrozen}
    \end{equation}
with all possible $v_{j,i}\in\{0,1\}$ and $\mathcal{A}^{(0)}\triangleq\varnothing$, where $\edd{\bigoplus}$ denotes XOR summation \edd{and $u_i\triangleq 0$ if $\mathcal{A}^{(i-1)}=\varnothing$ for any $i\in\mathcal{A}^c$}.
\end{definition}

\subsection{Related Decoding Algorithms}\label{sec:decoding}

\subsubsection{Successive Cancellation Decoding}

Let $c^{N}$ and $y^N$ be the transmitted and received words, respectively. \ac{SC} decod\edd{ing makes the} decision for the $i$-th bit-channel sequentially from $i=1$ to $i=N$  as follows. For $i\in\ed{\mathcal{A}^c}$\edd{,} set $\hat{u}_{i}$ to its (dynamic) frozen value. For $i\in\mathcal{A}$\edd{,} compute the soft message $\ell_i\left(\hat{u}_1^{i-1}\right)$ defined as
\begin{equation}\label{eq:soft_message}
    \ell_i\left(\hat{u}_1^{i-1}\right)\triangleq\log\frac{P_{U_i|Y^NU^{i-1}}(0|y^N,\hat{u}^{i-1})}{P_{U_i|Y^NU^{i-1}}(1|y^N,\hat{u}^{i-1})}
\end{equation}
assuming that the previous decisions $\hat{u}_{1}^{i-1}$ are correct and the frozen bits after $u_i$ are uniformly distributed. \edd{Then, it makes a hard decision as}
\begin{equation}\label{eq:dec_fnc}
	\hat{u}_{i}=\begin{cases}
	0 & \text{if }\ell_i\left(\hat{u}_1^{i-1}\right)\geq 0\\
	1 & \text{otherwise}.
	\end{cases}
\end{equation}
Any erroneous decision $\hat{u}_i\neq u_i$, $i\in\mathcal{A}$, cannot be corrected by \ac{SC} decoding and results in a frame error. In the following, we review techniques to overcome this problem.

\subsubsection{Successive Cancellation List Decoding}\label{sec:list_decoding}
\ac{SCL} decoding tracks several \ac{SC} decoding paths~\cite{tal2015list} in parallel. At each decoding phase $i\in\mathcal{A}$, instead of making a hard decision on $u_i$, two possible decoding paths are continued in parallel threads. The maximum number $2^{\edd{K}}$ of paths implements \ac{ML} decoding but with exponential complexity in $\edd{K}$. To limit complexity, one may keep up to $L$ paths at each phase. The reliability of decoding path $v^{i}$ is quantified by a \emph{\ac{PM}} defined as~\cite{balatsoukas2015llr}
\begin{align}
    M\left(v^{i}\right)&\triangleq -\log P_{U^i|Y^N}\left(v^i|y^N\right)\label{eq:score_scl}\\
    &=M\left(v^{i-1}\right)+\log\left(1+e^{-\left(1-2v_i\right)\ell_i\left(v_1^{i-1}\right)}\right)\\
    &\approx
\left\{\begin{aligned}
&M\left(v^{i-1}\right),~\text{if}~\text{sign}\left(\ell_i\left(v^{i-1}\right)\right)=1-2v_i\\
&M\left(v^{i-1}\right)+\left|\ell_i\left(v^{i-1}\right)\right|, ~\text{otherwise}
\end{aligned}\right.
    \label{eq:PM}
\end{align}
where \eqref{eq:PM} can be computed recursively using \ac{SC} decoding with $M\left(v^0\right)\triangleq 0$.
At the end of $N$-th decoding phase, a list $\mathcal{L}$ of paths is collected. Finally, the output is the bit vector minimizing the \ac{PM}:
\begin{equation}
    \hat{u}^N = \argmin_{v^{N}\in\mathcal{L}} M\left(v^{N}\right).
\end{equation}

\subsubsection{Flip Decoding}\label{sec:flip_decoding}
\ac{SCF} decoding \cite{afisiadis2014low} aims to correct the first erroneous bit decision by sequentially flipping the unreliable decisions. This procedure requires an error-detecting outer code, e.g., a \ac{CRC} code.

The \ac{SCF} decoder starts by performing \ac{SC} decoding for the inner code to generate the first estimate $v^N$. If $v^N$ passes the \ac{CRC} test, it is declared as the output $\hat{u}^N=v^N$. If not, then the \ac{SCF} algorithm attempts to correct the bit errors at most $T_\text{max}$ times. At the $t$-th attempt, $t\in[T_\text{max}]$, the decoder finds the index $i_t$ of the $t$-th least reliable decision in $v^N$ according to the amplitudes of the soft messages \eqref{eq:soft_message}. The \ac{SCF} algorithm restarts the \ac{SC} decoder by flipping the estimate $v_{i_t}$ to $v_{i_t}\oplus 1$. The CRC is checked after each attempt. This decoding process continues until the \ac{CRC} passes or $T_\text{max}$ is reached.

Introducing a bias term to account for the reliability of the previous decisions enhances the performance~\cite{chandesris2018dynamic}. The improved metric is calculated as
\begin{align}
    Q(i) = \left|\ell_i\left(v^{i-1}\right)\right| + \sum_{\substack{j\in\mathcal{A}^{(i)}}}\frac{1}{\alpha}\log\left(1+e^{-\alpha \left|\ell_j\left(v^{j-1}\right)\right|}\right)
\end{align}
where $\alpha>0$ is a scaling factor.

\ac{SCF} decoding can be generalized to flip multiple bit estimates at once, leading to \ac{DSCF} decoding~\cite{chandesris2018dynamic}. The reliability of the initial estimates $\tilde{u}_\mathcal{E}$, $\mathcal{E}\subseteq\mathcal{A}$, is described by
\begin{align}\label{eq:dscf}
    \!\!\!\!Q(\mathcal{E}) \!=\! \sum_{i\in\mathcal{E}}\left|\ell_i\left(v^{i-1}\right)\right| \!+\!\!\!\!\!\!\sum_{\substack{j\in\mathcal{A}^{(i_\text{max})}}}\!\frac{1}{\alpha}\log\!\left(1+e^{-\alpha \left|\ell_j\left(v^{j-1}\right)\right|}\right)
\end{align}
where $i_{\text{max}}$ is the largest element in $\mathcal{E}$. The set of flipping positions is chosen as the one minimizing the metric \eqref{eq:dscf} and is constructed progressively.

\subsubsection{Sequential Decoding}\label{sec:seq_RM}
We review two sequential decoding algorithms, namely \ac{SCS} decoding~\cite{NC12:seq,miloslavskaya2014sequential,trifonov2018score} and \ac{SC-Fano} decoding~\cite{jeong2019sc,arikan2019sequential}.

\ac{SCS} decoding stores the $\edd{D}$ most reliable paths (possibly) with different lengths and discards the rest whenever the stack is full. At each iteration, the decoder selects the most reliable path and creates two possible decoding paths based on this path. The winning word is declared once a path length becomes $N$. \edd{To limit its worst-case complexity similar to that of \ac{SCL} decoding with list size $L$, the decoding is limited to have at most $L$ visits each node in the decoding tree, which is finished in at most $LN$ node-visits.}
\ac{SC-Fano} decoding deploys a Fano search \cite{fano1963heuristic} that allows backward movement in the decoding tree and that uses a dynamic threshold. The dynamic threshold is initialized as $T=0$. During the Fano search, if one cannot find a path with score less than $T$ then the dynamic threshold is updated to $T+\Delta$, where $\Delta$ is called the threshold spacing and controls the performance vs. complexity tradeoff.

Sequential decoding compares paths of different lengths. The probabilities $P_{U^i|Y^N}\left(v^i|y^N\right)$, $v^i\in\{0,1\}^i$, however, cannot capture the effect of the path's length. A new score is introduced in\ed{~\cite{TM13} and used in~\cite{jeong2019sc}} to account for the expected error rate of the future bits as
\begin{align}
    S\left(v^i\right)&\triangleq -\log\frac{P_{U^i|Y^N}\left(v^i|y^N\right)}{\prod_{j=1}^{i} \left(1-p_{j}\right)}\label{eq:Score}\\
    &=M\left(v^{i}\right) + \sum_{j=1}^{i}\log\left(1-p_{j}\right)\label{eq:Score2}
\end{align}
where $p_{j}$ is the probability of the event that the first bit error occurred for $u_j$ in \ac{SC} decoding. The probabilities $p_{i}$ can be computed via Monte Carlo simulation\edddd{s}\cite{stolte2002rekursive,arikan2009channel} or they can be approximated via density evolution\cite{Mori:2009SIT} offline. In the following, one may generalize the score as
\begin{equation}
    S\left(v^i\right) = M\left(v^{i}\right) + b_i \label{eq:score_bias_term}
\end{equation}
where $b_i$ is called a bias term. We discuss in Section~\ref{sec:bias} how the bias term affects the proposed decoding algorithm.

\section{SC Ordered Search Decoding}
\label{sec:sc_os}

\ac{SCOS} uses the metrics (\ref{eq:PM}) and (\ref{eq:Score}). We first define
\begin{align}
    \overline{M}\left(v^i\right)&\triangleq M\left(v^{i-1}\overline{v_i}\right) = -\log P_{U^i|Y^N}\left(v^{i-1}\overline{v_i}\left|y^N\right.\right)\label{eq:PM_SCOS}\\
    \overline{S}\left(v^i\right)&\triangleq \overline{M}\left(v^i\right)+ \mcc{b_i}\label{eq:Score_SCOS}
\end{align}
\mcc{where \ed{$\Bar{v}_i\triangleq v_i\edd{\oplus} 1$, \eddd{$v^{i-1}v_i\triangleq v^i$} and} one may choose $b_i$ as the second term in the \ac{RHS} of \eqref{eq:Score2}.} Figure~\ref{fig:scos_tree} shows an example of the decoding for $N=4$ and $K=4$. \ac{SCOS} starts by \ac{SC} decoding to provide an output $v^N$ as the current most-likely leaf, e.g., the black path $(0111)$ in Figure~\ref{fig:scos_tree}. This initial \ac{SC} decoding computes and also stores the \ac{PM} $\overline{M}\left(v^i\right)$ and the score $\overline{S}\left(v^i\right)$ associated with the flipped versions of the decisions $v_i$ for all $i\in\mathcal{A}$, e.g., illustrated as the red paths in Figure~\ref{fig:scos_tree}. Every index $i\in\mathcal{A}$ with $\overline{M}\left(v^i\right) < M(v^N)$ \eddd{constitutes} a flipping set\eddd{, i.e., e}ach set is a singleton at this stage. \edddd{All flipping sets are stored \eddddd{in a} min heap $\mathcal{L}$~\cite{atkinson1986min}} and each member is visited in ascending order according to its score.
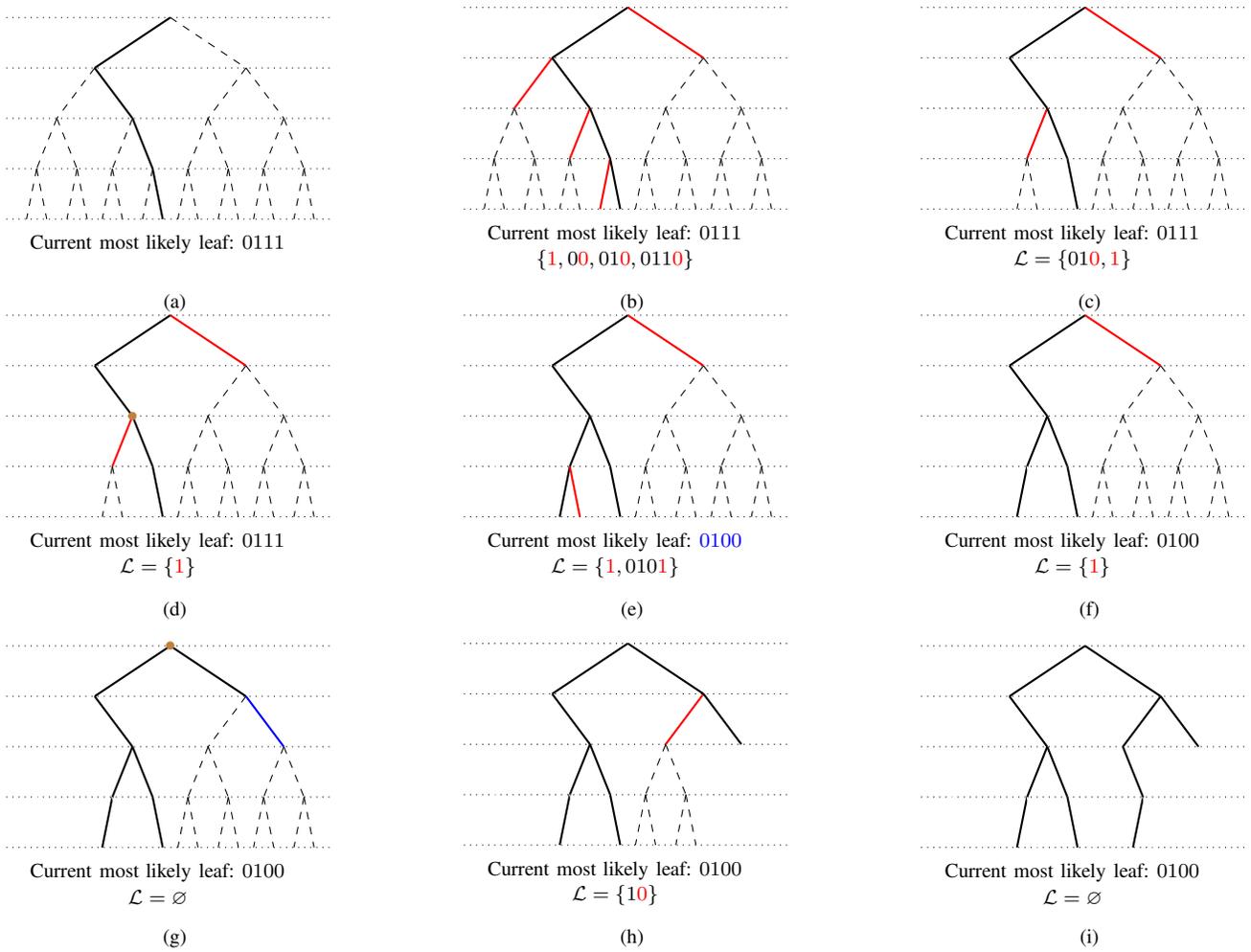
\begin{figure*}
	\centering
	\begin{subfigure}[b]{0.3\textwidth}
		\centering
		\begin{tikzpicture}[scale=0.7]
    \footnotesize
\draw [dashed] (0.5,4)--(0.1,3); \draw [dashed]  (0.5,4)--(0.9,3);
\draw [dashed] (2,4)--(1.6,3); \draw [thick] (2,4)--(2.4,3);
\draw [dashed] (3.5,4)--(3.1,3); \draw [dashed] (3.5,4)--(3.9,3);
\draw [dashed] (5,4)--(4.6,3); \draw [dashed] (5,4)--(5.4,3);

\draw [dashed] (1.25,5)--(0.5,4); \draw [thick]  (1.25,5)--(2,4);
\draw [dashed] (4.25,5)--(3.5,4); \draw [dashed] (4.25,5)--(5,4); 

\draw [thick] (2.75,6)--(1.25,5); \draw[dashed] (2.75,6)--(4.25,5);

\draw[dotted] (-0.5,2) to (6,2);
\draw[dotted] (-0.5,3) to (6,3);
\draw[dotted] (-0.5,4) to (6,4);
\draw[dotted] (-0.5,5) to (6,5);
\draw[dotted] (-0.5,6) to (6,6);

\draw [dashed] (0.1,3)--(-0.1,2); \draw [dashed]  (0.1,3)--(0.3,2);
\draw [dashed] (0.9,3)--(0.7,2); \draw [dashed]  (0.9,3)--(1.1,2);
\draw [dashed] (1.6,3)--(1.4,2); \draw [dashed]  (1.6,3)--(1.8,2);
\draw [dashed] (2.4,3)--(2.2,2); \draw [thick]  (2.4,3)--(2.6,2);
\draw [dashed] (3.1,3)--(2.9,2); \draw [dashed]  (3.1,3)--(3.3,2);
\draw [dashed] (3.9,3)--(3.7,2); \draw [dashed]  (3.9,3)--(4.1,2);
\draw [dashed] (4.6,3)--(4.4,2); \draw [dashed]  (4.6,3)--(4.8,2);
\draw [dashed] (5.4,3)--(5.2,2); \draw [dashed]  (5.4,3)--(5.6,2);


\node at (2.5,1.5) {Current most likely leaf: $0111$};
\node at (2.5,1) { };
\end{tikzpicture}
            \vspace{-3mm}
		\caption{}
		\label{fig:scos-tree1}
	\end{subfigure}
	\hfill
	\begin{subfigure}[b]{0.3\textwidth}
		\centering
		\begin{tikzpicture}[scale=0.7]
    \footnotesize
\draw [dashed] (0.5,4)--(0.1,3); \draw [dashed]  (0.5,4)--(0.9,3);
\draw [red,thick] (2,4)--(1.6,3); \draw [thick] (2,4)--(2.4,3);
\draw [dashed] (3.5,4)--(3.1,3); \draw [dashed] (3.5,4)--(3.9,3);
\draw [dashed] (5,4)--(4.6,3); \draw [dashed] (5,4)--(5.4,3);

\draw [red,thick] (1.25,5)--(0.5,4); \draw [thick]  (1.25,5)--(2,4);
\draw [dashed] (4.25,5)--(3.5,4); \draw [dashed] (4.25,5)--(5,4); 

\draw [thick] (2.75,6)--(1.25,5); \draw[red,thick] (2.75,6)--(4.25,5);

\draw[dotted] (-0.5,2) to (6,2);
\draw[dotted] (-0.5,3) to (6,3);
\draw[dotted] (-0.5,4) to (6,4);
\draw[dotted] (-0.5,5) to (6,5);
\draw[dotted] (-0.5,6) to (6,6);

\draw [dashed] (0.1,3)--(-0.1,2); \draw [dashed]  (0.1,3)--(0.3,2);
\draw [dashed] (0.9,3)--(0.7,2); \draw [dashed]  (0.9,3)--(1.1,2);
\draw [dashed] (1.6,3)--(1.4,2); \draw [dashed]  (1.6,3)--(1.8,2);
\draw [red,thick] (2.4,3)--(2.2,2); \draw [thick]  (2.4,3)--(2.6,2);
\draw [dashed] (3.1,3)--(2.9,2); \draw [dashed]  (3.1,3)--(3.3,2);
\draw [dashed] (3.9,3)--(3.7,2); \draw [dashed]  (3.9,3)--(4.1,2);
\draw [dashed] (4.6,3)--(4.4,2); \draw [dashed]  (4.6,3)--(4.8,2);
\draw [dashed] (5.4,3)--(5.2,2); \draw [dashed]  (5.4,3)--(5.6,2);

\node at (2.5,1.5) {Current most likely leaf: $0111$};
\node at (2.5,1) {$\left\{\textcolor{red}{1},0\textcolor{red}{0},01\textcolor{red}{0},011\textcolor{red}{0}\right\}$};
\end{tikzpicture}
            \vspace{-3mm}
		\caption{}
		\label{fig:scos-tree2}
	\end{subfigure}
	\hfill
	\begin{subfigure}[b]{0.3\textwidth}
		\centering
		\begin{tikzpicture}[scale=0.7]
    \footnotesize
\draw [white] (0.5,4)--(0.1,3); \draw [white]  (0.5,4)--(0.9,3);
\draw [red,thick] (2,4)--(1.6,3); \draw [thick] (2,4)--(2.4,3);
\draw [dashed] (3.5,4)--(3.1,3); \draw [dashed] (3.5,4)--(3.9,3);
\draw [dashed] (5,4)--(4.6,3); \draw [dashed] (5,4)--(5.4,3);

\draw [white] (1.25,5)--(0.5,4); \draw [thick]  (1.25,5)--(2,4);
\draw [dashed] (4.25,5)--(3.5,4); \draw [dashed] (4.25,5)--(5,4); 

\draw [thick] (2.75,6)--(1.25,5); \draw[red,thick] (2.75,6)--(4.25,5);

\draw[dotted] (-0.5,2) to (6,2);
\draw[dotted] (-0.5,3) to (6,3);
\draw[dotted] (-0.5,4) to (6,4);
\draw[dotted] (-0.5,5) to (6,5);
\draw[dotted] (-0.5,6) to (6,6);

\draw [white] (0.1,3)--(-0.1,2); \draw [white]  (0.1,3)--(0.3,2);
\draw [white] (0.9,3)--(0.7,2); \draw [white]  (0.9,3)--(1.1,2);
\draw [dashed] (1.6,3)--(1.4,2); \draw [dashed]  (1.6,3)--(1.8,2);
\draw [white] (2.4,3)--(2.2,2); \draw [thick]  (2.4,3)--(2.6,2);
\draw [dashed] (3.1,3)--(2.9,2); \draw [dashed]  (3.1,3)--(3.3,2);
\draw [dashed] (3.9,3)--(3.7,2); \draw [dashed]  (3.9,3)--(4.1,2);
\draw [dashed] (4.6,3)--(4.4,2); \draw [dashed]  (4.6,3)--(4.8,2);
\draw [dashed] (5.4,3)--(5.2,2); \draw [dashed]  (5.4,3)--(5.6,2);

\node at (2.5,1.5) {Current most likely leaf: $0111$};
\node at (2.5,1) {$\mathcal{L}=\left\{01\textcolor{red}{0},\textcolor{red}{1}\right\}$};
\end{tikzpicture}
            \vspace{-3mm}
		\caption{}
		\label{fig:scos-tree3}
	\end{subfigure}
	\vspace{5pt}
	\begin{subfigure}[b]{0.3\textwidth}
		\centering
		\begin{tikzpicture}[scale=0.7]
    \footnotesize
\draw [white] (0.5,4)--(0.1,3); \draw [white]  (0.5,4)--(0.9,3);
\draw [red,thick] (2,4)--(1.6,3); \draw [thick] (2,4)--(2.4,3);
\draw [dashed] (3.5,4)--(3.1,3); \draw [dashed] (3.5,4)--(3.9,3);
\draw [dashed] (5,4)--(4.6,3); \draw [dashed] (5,4)--(5.4,3);

\draw [white] (1.25,5)--(0.5,4); \draw [thick]  (1.25,5)--(2,4);
\draw [dashed] (4.25,5)--(3.5,4); \draw [dashed] (4.25,5)--(5,4); 

\draw [thick] (2.75,6)--(1.25,5); \draw[red,thick] (2.75,6)--(4.25,5);

\draw[dotted] (-0.5,2) to (6,2);
\draw[dotted] (-0.5,3) to (6,3);
\draw[dotted] (-0.5,4) to (6,4);
\draw[dotted] (-0.5,5) to (6,5);
\draw[dotted] (-0.5,6) to (6,6);

\draw [white] (0.1,3)--(-0.1,2); \draw [white]  (0.1,3)--(0.3,2);
\draw [white] (0.9,3)--(0.7,2); \draw [white]  (0.9,3)--(1.1,2);
\draw [dashed] (1.6,3)--(1.4,2); \draw [dashed]  (1.6,3)--(1.8,2);
\draw [white] (2.4,3)--(2.2,2); \draw [thick]  (2.4,3)--(2.6,2);
\draw [dashed] (3.1,3)--(2.9,2); \draw [dashed]  (3.1,3)--(3.3,2);
\draw [dashed] (3.9,3)--(3.7,2); \draw [dashed]  (3.9,3)--(4.1,2);
\draw [dashed] (4.6,3)--(4.4,2); \draw [dashed]  (4.6,3)--(4.8,2);
\draw [dashed] (5.4,3)--(5.2,2); \draw [dashed]  (5.4,3)--(5.6,2);

\node at (2,4) {\textcolor{brown}{$\bullet$}}; 

\node at (2.5,1.5) {Current most likely leaf: $0111$};
\node at (2.5,1) {$\mathcal{L}=\left\{\textcolor{red}{1}\right\}$};

\end{tikzpicture}
            \vspace{-3mm}
		\caption{}
		\label{fig:scos-tree4}
	\end{subfigure}
	\hfill
	\begin{subfigure}[b]{0.3\textwidth}
		\centering
		\begin{tikzpicture}[scale=0.7]
    \footnotesize

\draw[dotted] (-0.5,2) to (6,2);
\draw[dotted] (-0.5,3) to (6,3);
\draw[dotted] (-0.5,4) to (6,4);
\draw[dotted] (-0.5,5) to (6,5);
\draw[dotted] (-0.5,6) to (6,6);

\draw [white] (0.5,4)--(0.1,3); \draw [white]  (0.5,4)--(0.9,3);
\draw [thick] (2,4)--(1.6,3); \draw [thick] (2,4)--(2.4,3);
\draw [dashed] (3.5,4)--(3.1,3); \draw [dashed] (3.5,4)--(3.9,3);
\draw [dashed] (5,4)--(4.6,3); \draw [dashed] (5,4)--(5.4,3);

\draw [white] (1.25,5)--(0.5,4); \draw [thick]  (1.25,5)--(2,4);
\draw [dashed] (4.25,5)--(3.5,4); \draw [dashed] (4.25,5)--(5,4); 

\draw [thick] (2.75,6)--(1.25,5); \draw[red,thick] (2.75,6)--(4.25,5);

\draw [white] (0.1,3)--(-0.1,2); \draw [white]  (0.1,3)--(0.3,2);
\draw [white] (0.9,3)--(0.7,2); \draw [white]  (0.9,3)--(1.1,2);
\draw [thick] (1.6,3)--(1.4,2); \draw [red,thick]  (1.6,3)--(1.8,2);
\draw [white] (2.4,3)--(2.2,2); \draw [thick]  (2.4,3)--(2.6,2);
\draw [dashed] (3.1,3)--(2.9,2); \draw [dashed]  (3.1,3)--(3.3,2);
\draw [dashed] (3.9,3)--(3.7,2); \draw [dashed]  (3.9,3)--(4.1,2);
\draw [dashed] (4.6,3)--(4.4,2); \draw [dashed]  (4.6,3)--(4.8,2);
\draw [dashed] (5.4,3)--(5.2,2); \draw [dashed]  (5.4,3)--(5.6,2);


\node at (2.5,1.5) {Current most likely leaf: $\textcolor{blue}{0100}$};
\node at (2.5,1) {$\mathcal{L}=\left\{\textcolor{red}{1},010\textcolor{red}{1}\right\}$};

\end{tikzpicture}
            \vspace{-3mm}
		\caption{}
		\label{fig:scos-tree5}
	\end{subfigure}
	\hfill
	\begin{subfigure}[b]{0.3\textwidth}
		\centering
		\begin{tikzpicture}[scale=0.7]
    \footnotesize

\draw[dotted] (-0.5,2) to (6,2);
\draw[dotted] (-0.5,3) to (6,3);
\draw[dotted] (-0.5,4) to (6,4);
\draw[dotted] (-0.5,5) to (6,5);
\draw[dotted] (-0.5,6) to (6,6);

\draw [white] (0.5,4)--(0.1,3); \draw [white]  (0.5,4)--(0.9,3);
\draw [thick] (2,4)--(1.6,3); \draw [thick] (2,4)--(2.4,3);
\draw [dashed] (3.5,4)--(3.1,3); \draw [dashed] (3.5,4)--(3.9,3);
\draw [dashed] (5,4)--(4.6,3); \draw [dashed] (5,4)--(5.4,3);

\draw [white] (1.25,5)--(0.5,4); \draw [thick]  (1.25,5)--(2,4);
\draw [dashed] (4.25,5)--(3.5,4); \draw [dashed] (4.25,5)--(5,4); 

\draw [thick] (2.75,6)--(1.25,5); \draw[red,thick] (2.75,6)--(4.25,5);

\draw [white] (0.1,3)--(-0.1,2); \draw [white]  (0.1,3)--(0.3,2);
\draw [white] (0.9,3)--(0.7,2); \draw [white]  (0.9,3)--(1.1,2);
\draw [thick] (1.6,3)--(1.4,2); \draw [white]  (1.6,3)--(1.8,2);
\draw [white] (2.4,3)--(2.2,2); \draw [thick]  (2.4,3)--(2.6,2);
\draw [dashed] (3.1,3)--(2.9,2); \draw [dashed]  (3.1,3)--(3.3,2);
\draw [dashed] (3.9,3)--(3.7,2); \draw [dashed]  (3.9,3)--(4.1,2);
\draw [dashed] (4.6,3)--(4.4,2); \draw [dashed]  (4.6,3)--(4.8,2);
\draw [dashed] (5.4,3)--(5.2,2); \draw [dashed]  (5.4,3)--(5.6,2);


\node at (2.5,1.5) {Current most likely leaf: $0100$};
\node at (2.5,1) {$\mathcal{L}=\left\{\textcolor{red}{1}\right\}$};

\end{tikzpicture}
            \vspace{-3mm}
		\caption{}
		\label{fig:scos-tree6}
	\end{subfigure}
	\vspace{5pt}
	\begin{subfigure}[b]{0.3\textwidth}
		\centering
		\begin{tikzpicture}[scale=0.7]
    \footnotesize

\draw[dotted] (-0.5,2) to (6,2);
\draw[dotted] (-0.5,3) to (6,3);
\draw[dotted] (-0.5,4) to (6,4);
\draw[dotted] (-0.5,5) to (6,5);
\draw[dotted] (-0.5,6) to (6,6);

\draw [white] (0.5,4)--(0.1,3); \draw [white]  (0.5,4)--(0.9,3);
\draw [thick] (2,4)--(1.6,3); \draw [thick] (2,4)--(2.4,3);
\draw [dashed] (3.5,4)--(3.1,3); \draw [dashed] (3.5,4)--(3.9,3);
\draw [dashed] (5,4)--(4.6,3); \draw [dashed] (5,4)--(5.4,3);

\draw [white] (1.25,5)--(0.5,4); \draw [thick]  (1.25,5)--(2,4);
\draw [dashed] (4.25,5)--(3.5,4); \draw [thick,blue] (4.25,5)--(5,4); 

\draw [thick] (2.75,6)--(1.25,5); \draw[black,thick] (2.75,6)--(4.25,5);

\draw [white] (0.1,3)--(-0.1,2); \draw [white]  (0.1,3)--(0.3,2);
\draw [white] (0.9,3)--(0.7,2); \draw [white]  (0.9,3)--(1.1,2);
\draw [thick] (1.6,3)--(1.4,2); \draw [white]  (1.6,3)--(1.8,2);
\draw [white] (2.4,3)--(2.2,2); \draw [thick]  (2.4,3)--(2.6,2);
\draw [dashed] (3.1,3)--(2.9,2); \draw [dashed]  (3.1,3)--(3.3,2);
\draw [dashed] (3.9,3)--(3.7,2); \draw [dashed]  (3.9,3)--(4.1,2);
\draw [dashed] (4.6,3)--(4.4,2); \draw [dashed]  (4.6,3)--(4.8,2);
\draw [dashed] (5.4,3)--(5.2,2); \draw [dashed]  (5.4,3)--(5.6,2);

\node at (2.75,6) {\textcolor{brown}{$\bullet$}}; 

\node at (2.5,1.5) {Current most likely leaf: $0100$};
\node at (2.5,1) {$\mathcal{L}=\varnothing$};

\end{tikzpicture}
            \vspace{-3mm}
		\caption{}
		\label{fig:scos-tree7}
	\end{subfigure}
	\hfill
	\begin{subfigure}[b]{0.3\textwidth}
		\centering
		\begin{tikzpicture}[scale=0.7]
    \footnotesize

\draw[dotted] (-0.5,2) to (6,2);
\draw[dotted] (-0.5,3) to (6,3);
\draw[dotted] (-0.5,4) to (6,4);
\draw[dotted] (-0.5,5) to (6,5);
\draw[dotted] (-0.5,6) to (6,6);

\draw [white] (0.5,4)--(0.1,3); \draw [white]  (0.5,4)--(0.9,3);
\draw [thick] (2,4)--(1.6,3); \draw [thick] (2,4)--(2.4,3);
\draw [dashed] (3.5,4)--(3.1,3); \draw [dashed] (3.5,4)--(3.9,3);
\draw [white] (5,4)--(4.6,3); \draw [white] (5,4)--(5.4,3);

\draw [white] (1.25,5)--(0.5,4); \draw [thick]  (1.25,5)--(2,4);
\draw [red,thick] (4.25,5)--(3.5,4); \draw [thick] (4.25,5)--(5,4); 

\draw [thick] (2.75,6)--(1.25,5); \draw[black,thick] (2.75,6)--(4.25,5);

\draw [white] (0.1,3)--(-0.1,2); \draw [white]  (0.1,3)--(0.3,2);
\draw [white] (0.9,3)--(0.7,2); \draw [white]  (0.9,3)--(1.1,2);
\draw [thick] (1.6,3)--(1.4,2); \draw [white]  (1.6,3)--(1.8,2);
\draw [white] (2.4,3)--(2.2,2); \draw [thick]  (2.4,3)--(2.6,2);
\draw [dashed] (3.1,3)--(2.9,2); \draw [dashed]  (3.1,3)--(3.3,2);
\draw [dashed] (3.9,3)--(3.7,2); \draw [dashed]  (3.9,3)--(4.1,2);
\draw [white] (4.6,3)--(4.4,2); \draw [white]  (4.6,3)--(4.8,2);
\draw [white] (5.4,3)--(5.2,2); \draw [white]  (5.4,3)--(5.6,2);


\node at (2.5,1.5) {Current most likely leaf: $0100$};
\node at (2.5,1) {$\mathcal{L}=\left\{1\textcolor{red}{0}\right\}$};

\end{tikzpicture}
            \vspace{-3mm}
		\caption{}
		\label{fig:scos-tree8}
	\end{subfigure}
	\hfill
	\begin{subfigure}[b]{0.3\textwidth}
		\centering
		\begin{tikzpicture}[scale=0.7]
    \footnotesize

\draw[dotted] (-0.5,2) to (6,2);
\draw[dotted] (-0.5,3) to (6,3);
\draw[dotted] (-0.5,4) to (6,4);
\draw[dotted] (-0.5,5) to (6,5);
\draw[dotted] (-0.5,6) to (6,6);

\draw [white] (0.5,4)--(0.1,3); \draw [white]  (0.5,4)--(0.9,3);
\draw [thick] (2,4)--(1.6,3); \draw [thick] (2,4)--(2.4,3);
\draw [white] (3.5,4)--(3.1,3); \draw [thick] (3.5,4)--(3.9,3);
\draw [white] (5,4)--(4.6,3); \draw [white] (5,4)--(5.4,3);

\draw [white] (1.25,5)--(0.5,4); \draw [thick]  (1.25,5)--(2,4);
\draw [thick] (4.25,5)--(3.5,4); \draw [thick] (4.25,5)--(5,4); 

\draw [thick] (2.75,6)--(1.25,5); \draw[black,thick] (2.75,6)--(4.25,5);

\draw [white] (0.1,3)--(-0.1,2); \draw [white]  (0.1,3)--(0.3,2);
\draw [white] (0.9,3)--(0.7,2); \draw [white]  (0.9,3)--(1.1,2);
\draw [thick] (1.6,3)--(1.4,2); \draw [white]  (1.6,3)--(1.8,2);
\draw [white] (2.4,3)--(2.2,2); \draw [thick]  (2.4,3)--(2.6,2);
\draw [white] (3.1,3)--(2.9,2); \draw [white]  (3.1,3)--(3.3,2);
\draw [thick] (3.9,3)--(3.7,2); \draw [white]  (3.9,3)--(4.1,2);
\draw [white] (4.6,3)--(4.4,2); \draw [white]  (4.6,3)--(4.8,2);
\draw [white] (5.4,3)--(5.2,2); \draw [white]  (5.4,3)--(5.6,2);


\node at (2.5,1.5) {Current most likely leaf: $0100$};
\node at (2.5,1) {$\mathcal{L}=\varnothing$};

\end{tikzpicture}
            \vspace{-3mm}
		\caption{}
		\label{fig:scos-tree9}
	\end{subfigure}
	\caption{(a) Initial SC decoding outputs $v^N = (0111)$ with the corresponding \ac{PM} $M(v^N)$. (b) During the initial SC decoding, the \acp{PM} and scores are computed for branch nodes $\left\{\textcolor{red}{1},0\textcolor{red}{0},01\textcolor{red}{0},011\textcolor{red}{0}\right\}$. (c) The branch nodes with \acp{PM} larger than that of the current most likely leaf are pruned, e.g., we have $M(0\textcolor{red}{0}), M(011\textcolor{red}{0}) > M(0111)$. Suppose \ed{also} that $S(01\textcolor{red}{0})<S(\textcolor{red}{1})$. Then, \edddd{$\mathcal{L}$ stores} all branch nodes with \acp{PM} smaller than that of current most likely leaf, \edddd{where $\mathcal{L}$ is a min heap according to the scores of its members.} (d) The \edddd{candidate with smallest score} is popped from the \edddd{heap} and the decoder returns to the deepest (or nearest) common node. (e) The decision is flipped and SC decoding continues. During decoding, the \edddd{heap} $\mathcal{L}$ and the current most likely leaf are updated. (f) The branch nodes with \acp{PM} larger than that of the current most likely leaf are pruned as in (c) (in this case, a leaf node is removed). (g) Repeat the procedure as in step (d), where we assume $M\left({1}\textcolor{blue}{1}\right)>M\left(0100\right)$. (h) The heap $\mathcal{L}$ is updated when the branch $\left({1}\textcolor{blue}{1}\right)$ was visited. (i) The decoder examines the last member of the heap $\mathcal{L}$ and pops it from $\mathcal{L}$. After reaching the $N$-th decoding phase, suppose that there is no branch node left, which has a smaller \ac{PM} than that of the current most likely leaf, i.e., $\mathcal{L}=\varnothing$. The current most likely leaf is declared as the decision $\hat{u}^N$.}
	\label{fig:scos_tree}
\end{figure*}

\edddd{Let $\mathcal{E}$ be the flipping set with the smallest score in the \edddd{heap} $\mathcal{L}$ and} let index $j\in[N]$ be the deepest common node of the current most-likely leaf and the branch node defined by $\mathcal{E}$ in the decoding tree (see the brown dot in Figure~\ref{fig:scos_tree}\ed{(d)}). The decoder now flips the decision $v_j$ and \ac{SC} decoding continues. The set $\mathcal{E}$ is popped from the \edddd{heap} $\mathcal{L}$. The \acp{PM}~(\ref{eq:PM}) and scores~(\ref{eq:Score}) are calculated again for the flipped versions for decoding phases with $i>j$, $i\in\mathcal{A}$, and the \edddd{heap} $\mathcal{L}$ is enhanced by new flipping sets progressively (similar to \cite{chandesris2018dynamic}). The branch node, including all of its child nodes, is discarded if at any decoding phase its \ac{PM} exceeds that of the current most-likely leaf, i.e., $M(v^N)$.\footnote{This pruning method is similar to the adaptive skipping rule proposed in~\cite{wu2007soft} for ordered-statistics decoding~\cite{D74,fossorier}.} Such a branch cannot output the \ac{ML} decision, since for any valid path $v^i$ the \ac{PM}~(\ref{eq:PM}) is non-decreasing for the next stage, i.e., we have
\begin{align}\label{eq:pm_increase}
	M\left(v^{i}\right) \leq M\left(v^{i+1}\right), \forall v_{i+1}\in\{0,1\}.
\end{align}
For instance, suppose that $M(11)>M(0111)$ in Figure~\ref{fig:scos_tree}(g). Then any path $\tilde{v}^N$ with $\tilde{v}^2 = (1,1)$ cannot be the \ac{ML} decision; hence, it is pruned. If a leaf with lower \ac{PM} is found, then it replaces the current most-likely leaf. The procedure is repeated until one cannot find a more reliable path by flipping decisions, i.e., until $\mathcal{L}=\varnothing$. Hence, \ac{SCOS} decoding implements an \ac{ML} decoder.

\subsection{\mcc{Detailed} Description}\label{sec:pseudo_scos}
This section \mcc{provides} \edd{the details of the proposed \ac{SCOS} decoding with the pseudo codes}\edd{, where a simulation code is provided in \cite{scml}. In the following}, we use type-writer font for the data structures (with an exception for sets) and $1$-based indexing arrays. \mcc{The required data structures together with their size are listed} in Table~\ref{tab:memory_SCOS}. \edd{As we explain the algorithms, we will revisit the relevant data structure from the table. We start with} arrays $\texttt{L}$ and $\texttt{C}$, \edd{which} \ed{contain \acp{LLR} and hard decisions, respectively.} \edd{Recall that there are $\log_2N+1$ layers in a polar code graph and both $\texttt{L}$ and $\texttt{C}$ store $N$ elements in each layer (in contrast to~\cite{tal2015list} where \edd{in total} only $2N-1$ elements are stored)} since we reuse some decoding paths to decrease the computational complexity (similar to \ac{SC-Fano} decoding). \edd{The entry in position $(i,j)$ of array $\texttt{L}$ \edddd{($\texttt{C}$)} is denoted as $\texttt{L}\left[i,j\right]$ \edddd{($\texttt{C}\left[i,j\right]$)}, which is calculated via Algorithm \ref{alg:recursivelyCalcL} \edddd{(\ref{alg:recursivelyCalcC})}. These routines, namely recursivelyCalcL and recursivelyCalcC, are the \ac{LLR}-based versions\cite{balatsoukas2015llr} of \cite[Alg. 3]{tal2015list} and \cite[Alg. 4]{tal2015list}, respectively, and provided as Algorithm \ref{alg:recursivelyCalcL} and Algorithm \ref{alg:recursivelyCalcC} in the appendix for completeness. We also \edddd{name indices} $\lambda$ and $\phi$ \edddd{as layer and phase, respectively,} by adopting the convention of \cite{tal2015list}\edddd{, which are integer-valued inputs to Algorithms \ref{alg:recursivelyCalcL} and \ref{alg:recursivelyCalcC}. U}nlike \cite{tal2015list}, \edddd{the layer and phase} satisfy $1\leq\lambda\leq \log_2N+1$ and $1\leq\phi < 2^\lambda$ due to $1$-based indexing.}
\begin{table*}[t]
\footnotesize
	\centering
	\caption{\mcc{Data structures for \ac{SCOS} decoding.}}
	\setlength\extrarowheight{0pt}
	\setlength{\tabcolsep}{18pt}
	\begin{tabular}{|c|c|c|c|}
		\hline
		name & size & data type & description\\
		\hline
		\hline
		$\texttt{L}$ & $\left(\log_2N+1\right)\times N$ & float & \ac{LLR} \\
		\hline
		$\texttt{C}$ & $\left(\log_2N+1\right)\times N$ & binary & hard decision \\
        \hline
		$\texttt{F}$ & $1$ & $\left\langle \text{set},\text{float},\text{float} \right\rangle$ & structure of a flipping set \\
		\hline
        $\mathcal{E}$, $\mathcal{E}_\texttt{p}$ & \ed{$\leq K$} & \ed{integer (set of indices)}& flipping set\\
		\hline
		$\mathcal{L}$ & $\leq \eta$ & type of $\texttt{F}$ & \edddd{heap} of flipping structures \\
		\hline
		$\hat{\texttt{u}}$, $\texttt{v}$ & $N$ & binary & decoding path \\
        \hline
		$\texttt{b}$ & $N$ & float & precomputed bias term\\
		\hline
		$\texttt{M}$, $\overline{\texttt{M}}$, $\overline{\texttt{S}}$ & $N$ & float & metric \\
		\hline
		$\texttt{M}_{\texttt{cml}}$ & $1$ & float & \ac{PM} of the current most likely leaf \\
		\hline
	\end{tabular}
	\label{tab:memory_SCOS}
\end{table*}

Flipping set structures, denoted by $\texttt{F}$, are triplets containing a set of integer indices (flipping set $\mathcal{E}$), a \ac{PM} and a score. The \edddd{heap} $\mathcal{L}$ contains multiple flipping structures $\texttt{F}=\left\langle\mathcal{E},\overline{\texttt{M}}_\mathcal{E},\overline{\texttt{S}}_\mathcal{E}\right\rangle$,\footnote{Observe that the heap $\mathcal{L}$ in Figure~\ref{fig:scos_tree} is slightly different for simplicity.} where $\overline{\texttt{M}}_\mathcal{E}$ and $\overline{\texttt{S}}_\mathcal{E}$ are the respective \ac{PM} and the score associated to the flipping set $\mathcal{E}$, as defined in \eqref{eq:PM_SCOS} and \eqref{eq:Score_SCOS}. The size of $\mathcal{L}$ is constrained by a user-defined parameter $\eta$, \eddd{which will define the space complexity of the decoder}. Given two flipping sets, namely $\mathcal{E}$ and $\mathcal{E}_p$, Algorithm~\ref{alg:FindStartIndex} is the procedure used to find the decoding stage to which the decoder should return, i.e., the deepest common node as illustrated in Figure~\ref{fig:scos_tree}(d).
\begin{algorithm}[t]
        \footnotesize
	\SetNoFillComment
	\DontPrintSemicolon
	\SetKwInOut{Input}{Input}\SetKwInOut{Output}{Output}
	\Input{\mcc{flipping sets $\mathcal{E}$ and $\mathcal{E}_\texttt{p}$}}
	\Output{first different index}	
	\BlankLine
	\For{$i=1,2,\dots,N$}{
		\If{$\left(i\in\mathcal{E}\right) \oplus \left(i\in\mathcal{E}_\texttt{p} \right)$}{
			\Return{$i$}
		}
	}
	\caption{$\text{FindStartIndex}\left(\mathcal{E},\mathcal{E}_\texttt{p}\right)$}
	\label{alg:FindStartIndex}
\end{algorithm}

Algorithm~\ref{alg:SCDec} is generalized \ac{SC} decoding, which can start \ac{SC} decoding at any decoding phase and continue decoding until a termination criterion is satisfied. Then, it returns the phase at which the decoding is terminated. The modifications compared to the original SC decoding are highlighted as blue in the pseudo code. Before their detailed descriptions, we recall data structures needed from Table \ref{tab:memory_SCOS}. The notation $\texttt{v}\left[i\right]$ refers to the $i$-th entry of an array $\texttt{v}$, \edd{where binary vectors $\texttt{v}$ and $\hat{\texttt{u}}$ are the currently processed path and the current most-likely one, respectively}. \edd{Unless otherwise stated, t}he entries of \edd{vector $\texttt{b}$} are computed offline via
\begin{align}\label{eq:bias_term}
	\texttt{b}[i] = \sum_{j=1}^{i}\log\left(1-p_{j}\right),~i\in[N].
\end{align}
The entries of length-$N$ vectors $\texttt{M}$ and $\overline{\texttt{M}}$ correspond to \acp{PM} along traversed paths and the flipped versions, respectively. Vector $\overline{\texttt{S}}$ contains the scores used for the search schedule of the proposed decoder. The modified \ac{SC} decoding takes as input an integer $i_\text{start}\in[N]$ and a flipping set $\mathcal{E}$ and outputs another index $i_\text{end}$ such that $i_\text{start} < i_\text{end}\leq N$. Along the way, the algorithm updates the vectors containing \acp{PM} and scores, namely $\texttt{M}$ and $\overline{\texttt{M}}$ and $\overline{\texttt{S}}$, where the details are itemized as follows.
\begin{itemize}
    \item \edd{The standard subroutine HardDec (\texttt{lines 8} and \texttt{10}) takes a real-valued \ac{LLR} as the input and returns a decision according to \eqref{eq:dec_fnc}. In addition, CalcPM (\texttt{lines 12} and \texttt{14}) takes a real-valued \ac{PM}, a binary decision and a real-valued \ac{LLR} as inputs and updates the \ac{PM} using \eqref{eq:PM}.}
    \item One can start at any decoding phase $i_\text{start}$ \ed{with no additional computational cost} \edd{(\texttt{line 2})}.
    \item \edd{The $i$-th entry of vector $\texttt{M}$ is updated in each decoding phase $i$} (\texttt{line 14}).
    \item The decisions are flipped at the decoding phases \edd{corresponding to the current flipping set, i.e., if $i\in\mathcal{E}$} (\texttt{lines 7-8}).
    \item \edd{The \acp{PM} and the scores of the potential flipping sets are computed for decoding phases after the largest one in the current flipping sets, i.e., $\overline{\texttt{M}}[i]$ and $\overline{\texttt{S}}[i]$ with $i\in\mathcal{A}$ and $i>\text{maximum}\left(\mathcal{E}\right)$} (\texttt{lines 11-13}).
    \item If a more likely leaf (i.e., a path of length-$N$ \edd{with smaller \ac{PM}}) is found, \edd{$\hat{\texttt{u}}$ and $\texttt{M}_\texttt{cml}$ are updated and the decoding phase $N$ is returned as $i_\text{end}$ (\texttt{lines 20-24})}.
    \item \edd{For any $i$, if} \ac{PM} $\texttt{M}[i]$ is larger than $\texttt{M}_\texttt{cml}$, stop SCDec function and return the current phase $i$ as $i_\text{end}$ (\texttt{line 15-16}).
    \item \ed{If $i$-th bit is dynamically frozen, then the computation of $\texttt{v}[i]$ follows the constraints, i.e., using \ac{RHS} of \eqref{eq:dfrozen} where the coefficients $v_{j,i}$ are specified by the construction \edd{(\texttt{line 5})}.}
\end{itemize}
\begin{algorithm}[t]
        \footnotesize
	\SetNoFillComment
	\DontPrintSemicolon
	\SetKwInOut{Input}{Input}\SetKwInOut{Output}{Output}
	\Input{\textcolor{blue}{start index $i_\text{start}$, flipping set $\mathcal{E}$}}
	\Output{\textcolor{blue}{end index $i_\text{end}$}}
	\BlankLine
	$m=\log_2N$\\
	\For{$i=\textcolor{blue}{i_\text{start}},\dots,N$}{
		$\text{recursivelyCalcL}\left(m+1,i-1\right)$\\
		\eIf{$i\notin\mathcal{A}$}{
			$\texttt{v}\left[i\right] = 0$\tcp*[l]{\textcolor{blue}{compute $\texttt{v}\left[i\right]$ if dynamic frozen}}
		}{
			\textcolor{blue}{\eIf{$i\in\mathcal{E}$}{
					$\texttt{v}\left[i\right] = \text{HardDec}\left(\texttt{L}\left[m+1,i\right]\right)\oplus 1$
				}{
					\textcolor{black}{$\texttt{v}\left[i\right] = \text{HardDec}\left(\texttt{L}\left[m+1,i\right]\right)$}
			}}
			\textcolor{blue}{\If{$i>\text{maximum}\left(\mathcal{E}\right)$}{
					$\overline{\texttt{M}}\left[i\right]=\text{CalcPM}\left(\texttt{M}\left[i-1\right],\texttt{v}\left[i\right]\oplus 1,\texttt{L}\left[m+1,i\right]\right)$\\
					$\overline{\texttt{S}}\left[i\right]=\overline{\texttt{M}}\left[i\right]+\texttt{b}\left[i\right]$
			}}
		}
		\textcolor{blue}{$\texttt{M}\left[i\right]=\text{CalcPM}\left(\texttt{M}\left[i-1\right],\texttt{v}\left[i\right],\texttt{L}\left[m+1,i\right]\right)$\\
			\If{$\texttt{M}\left[i\right] \geq \texttt{M}_\texttt{cml}$}{
				\Return{$i$}
		}}
		$\texttt{C}\left[m+1,i\right]=\texttt{v}\left[i\right]$\\
		\If{$i\mod 2 = 0$}{
			$\text{recursivelyCalcC}\left(m+1,i-1\right)$
		}
	}
	\textcolor{blue}{\If{$\texttt{M}\left[N\right] < \texttt{M}_\texttt{cml}$}{
			$\texttt{M}_\texttt{cml}=\texttt{M}\left[N\right]$\\
			\For{$i=1,2,\dots,N$}{
				$\hat{\texttt{u}}\left[i\right]=\texttt{v}\left[i\right]$
			}
			\Return{$N$}
	}}
	\caption{$\text{SCDec}\left(\textcolor{blue}{i_\text{start},\mathcal{E}}\right)$}
	\label{alg:SCDec}
\end{algorithm}

Algorithm~\ref{alg:SCOS} is the main loop of \ac{SCOS} \edd{decoding. Naturally, the \edddd{heap} of flipping structures and the previous flipping set are initialized as null and the \ac{PM} of the current most-likely leaf as $+\infty$}. After the initial \ac{SC} decoding (\texttt{line 4}), $\texttt{M}_\texttt{cml}$ is updated to the \ac{PM} of the \ac{SC} estimate. Then, a tree search is performed \edd{in order to find the most-likely estimate (\texttt{lines 8-16}), where the candidates are ordered by their scores}. Many sub-trees are pruned \edd{thanks to} the threshold $\texttt{M}_\texttt{cml}$ \edd{(\texttt{lines 6-7 and 14-15})}, i.e., the \ac{PM} of the current most likely leaf. The stopping condition of the \enquote{while loop} (\texttt{line 8} with $\mathcal{L}=\varnothing$) implies that the most likely codeword is found, i.e., there cannot be any other codeword with a smaller \ac{PM}. The estimate \edd{with \ac{PM}} $\texttt{M}_\texttt{cml}$ is output as the decision \edd{(\texttt{line 17})}.
\begin{algorithm}[t]
        \footnotesize
	\SetNoFillComment
	\DontPrintSemicolon
	\SetKwInOut{Input}{Input}\SetKwInOut{Output}{Output}
	\Input{\ac{LLR}s $\ell^N$}
	\Output{$\hat{\texttt{u}}$}
	\BlankLine
	$\mathcal{L}=\varnothing, \mathcal{E}_\texttt{p}=\varnothing, \texttt{M}_\texttt{cml}=+\infty$\\
	\For{$i=1,2,\dots,N$}{
		$\texttt{L}\left[1,i\right]=\ell_i$
	}
	$\text{SCDec}\left(1,\varnothing\right)$\\
	\For{$i=1,2,\dots,N$}{
		\If{$i\in\mathcal{A}$ and $\overline{\texttt{M}}\left[i\right]<\texttt{M}_\texttt{cml}$}{
			$\text{Insert\edddd{Heap}}\left(\left\langle\{i\},\overline{\texttt{M}}\left[i\right],\overline{\texttt{S}}\left[i\right]\right\rangle\right)$\\
		}
	}
	\While{$\mathcal{L}\neq\varnothing$}{
		$\left\langle \mathcal{E},\overline{\texttt{M}}_\mathcal{E},\overline{\texttt{S}}_\mathcal{E}\right\rangle = \text{pop\edddd{Min}}\left(\mathcal{L}\right)$\\
		\If{$\overline{\texttt{M}}_{\mathcal{E}}< \texttt{M}_\texttt{cml}$}{
			$i_\text{start}=\text{FindStartIndex}\left(\mathcal{E},\mathcal{E}_\texttt{p}\right)$\\
			$i_\text{end}=\text{SCDec}\left(i_\text{start},\mathcal{E}\right)$\\
			\For{$i=\text{maximum}\left(\mathcal{E}\right)+1,\dots,i_\text{end}$}{
				\If{$i\in\mathcal{A}$ and $\overline{\texttt{M}}\left[i\right]<\texttt{M}_\texttt{cml}$}{
					$\text{Insert\edddd{Heap}}\left(\left\langle\mathcal{E}\cup\{i\},\overline{\texttt{M}}\left[i\right],\overline{\texttt{S}}\left[i\right]\right\rangle\right)$\\
				}
			}
			$\mathcal{E}_\texttt{p} =\mathcal{E}$
		}
	}
	\Return{$\hat{\texttt{u}}$}
	\caption{$\text{SCOS}\left(\ell^N\right)$}
	\label{alg:SCOS}
\end{algorithm}
\eddd{\begin{remark}
    Observe that each member in the \edddd{heap} $\mathcal{L}$, where $|\mathcal{L}|\leq\eta$, stores a set of integers with maximum size of $K$ and two float metrics. Hence, $\mathcal{L}$ stores at most $K\eta$ integers and $2\eta$ floats. In addition, recall that the arrays $\texttt{L}$ and $\texttt{C}$ store $N\log_2N + N$ elements each in contrast to $\eta\times\left(2N-1\right)$, which is the case, e.g., in \ac{SCS} decoding with stack size $D=\eta$~\cite{miloslavskaya2014sequential,trifonov2018score}. Other data structures listed in Table \ref{tab:memory_SCOS} are of size at most $N$ each. In total, \ac{SCOS} stores at most $N\log_2N + 5N + 2\eta + 1$ floats, $NR\eta$ integers, and $N\log_2N + 3N$ bits, where $R$ is the code rate.
\end{remark}}

\section{Complexity \edddd{and Performance} Considerations}
\label{sec:comp}
\edd{We adopt \emph{number of node-visits} in the decoding tree as a proxy for the complexity. \edddd{A node visit occurs each time line 3 is executed in Algorithm \ref{alg:SCDec}, i.e., each time a node is visited in decoding tree, which is, for instance, provided for the case of $N=4$ in Figure \ref{fig:scos_tree}.} \ed{Note that this \edd{does} not refer to the exact complexity; however, it still provides a very good proxy \cite[Sec. 4.2]{YFV21}, which \edd{is} used very often \edd{in} prior works, see, e.g., \cite{arikan2019sequential,jeong2019sc,MMQ20,Moradi21}, among many} \edd{other references.\footnote{\edd{In our simulations, we count the number of arithmetic operations as well, which will be provided later in Sections \ref{sec:comp_perf} and \ref{sec:numerical} for various comparisons.}}} To this end, let $\lambda\left(y^N\right)$ be the number of node-visits in the decoding tree by \ac{SCOS} decoding normalized by block length for channel output $y^N$. Similar to other sequential decoders~\cite{jacobs1967lower}, it is a \ac{RV} defined as $\Lambda\triangleq \lambda\left(Y^N\right)$. On the contrary,} the number $\lambda_\text{SC}\edd{(y^N)}$ of \edd{normalized} node-visits for \ac{SC} decoding is simply \edd{$1$} \mcc{independent of the channel output $y^N$}.

\subsection{\edd{Average Number of Node-Visits for ML Decoding}}
\label{sec:compML}
\edd{In the following, we are interested in the average behaviour of $\Lambda$ when there is no limit in the number of node-visits. To this end,} consider the set of partial input sequences $v^i$, $i\in[N]$, with a smaller \ac{PM} than the \ac{ML} decision $\hat{u}^N_{\text{ML}}$. \ed{Observe that there are $i$ node-visits for \ac{SC} decoding for any decoding path $v^i$.}
\begin{definition}\label{def:setV}
    For the channel output $y^N$ and the binary sequence $v^N\in\mathcal{X}^N$, define the set
    \begin{equation}
        \mathcal{V}\left(v^N,y^N\right)\triangleq\bigcup_{i=1}^N\left\{u^i\in\{0,1\}^i: M\left(u^i\right) \leq M\left(v^N\right)\right\}.\label{eq:set_def}
    \end{equation}
\end{definition}
\begin{lemma}\label{lem:comp}
    We have  
    \begin{equation}
        \edd{N}\lambda\left(y^N\right)\geq \left|\mathcal{V}\left(\hat{u}^N_{\text{ML}}\left(y^N\right),y^N\right)\right|\label{eq:complexity_realization}
    \end{equation}
    and the \edd{\ac{ANV} normalized to that of SC decoding} is lower bounded as
    \begin{equation}
        \mathbb{E}\left[\Lambda\right]\geq \frac{1}{N}\mathbb{E}\left[\left|\mathcal{V}\left(\hat{u}^N_{\text{ML}}\left(Y^N\right),Y^N\right)\right|\right].\label{eq:complexity_average}
    \end{equation}
\end{lemma}
\begin{proof}
\edddd{Observe that each member of set $\mathcal{V}\left(v^N,y^N\right)$ corresponds to a node in the decoding tree. Then, i}nequality~\eqref{eq:complexity_realization} follows from~\eqref{eq:set_def} by replacing $v^N$ with the \ac{ML} decision $\hat{u}^N_{\text{ML}}$ and \edddd{observing that \ac{SCOS} decoding needs to visit each node with a \ac{PM} smaller than or equal to that of the \ac{ML} decision}. Since~\eqref{eq:complexity_realization} is valid for any $y^N$, the bound (\ref{eq:complexity_average}) follows.
\end{proof}

\edd{From now on, we refer to $\mathbb{E}\left[\Lambda\right]$ as \ac{ANV} by keeping in mind that it is normalized to the block length. Figure \ref{fig:scos_pac} provides performance\eddd{, the \ac{ANV} and the histogram for the node visits} vs. \ac{SNR} (in $E_b/N_0$, where $E_b$ is here the energy per information bit and $N_0$ is the single-sided noise power spectral density) over the \ac{biAWGN} channel\cite[Sec. 4]{Richardson:2008:MCT:1795974} for the $(128, 64)$ \ac{PAC} code~\cite{arikan2019sequential} under \ac{SCOS} decoding where $\eta=\infty$.} Observe that the lower bound on the \ac{ANV} given by \eqref{eq:complexity_average} is validated\edddd{\footnote{\edddd{After finding the \ac{PM} of \ac{ML} decision for each transmission, the number of nodes in the decoding tree with lower \ac{PM} than that of the \ac{ML} decision is counted via a modified \ac{SCOS} decoding, which is introduced in Section \ref{sec:SCOS_max_PM}. Then, its average provides the \ac{RHS} of Eq. \eqref{eq:complexity_average}.}}} and is tight for high \ac{SNR}. However, the bound appears to be loose at low \ac{SNR} values mainly for two reasons: (i) usually the initial \ac{SC} decoding estimate $v^N$ is not the \ac{ML} decision and extra nodes in the difference set $\mathcal{V}\left(v^N,y^N\right)\setminus\mathcal{V}\left(v^N_{\text{ML}},y^N\right)$ are visited and (ii) \ac{SCOS} decoding may visit the same node multiple times and this cannot be tracked by a set definition. The histogram for the node visits, where the intervals are given as integer multiples of node visit of SC decoding, reveals the efficiency of the proposed decoder for high SNR regime. In particular, the probability that SCOS decoding needs a number of node visits larger than that of $8$ times of SC decoding to guarantee returning \ac{ML} decision is roughly $2\times 10^{-4}$ when $E_b/N_0=4$ dB.
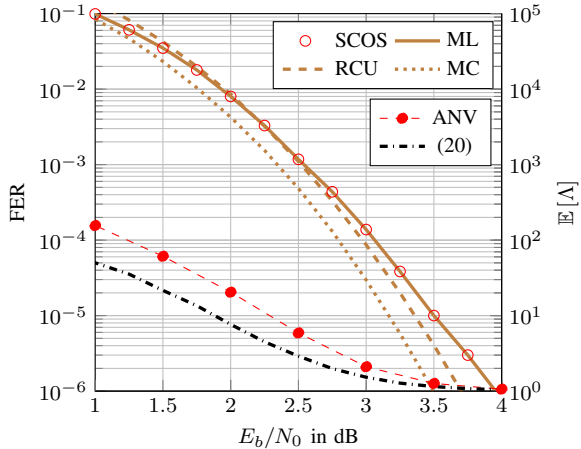
\begin{figure}
\centering
\begin{subfigure}{0.49\textwidth}
    \begin{tikzpicture}[scale=1]
\footnotesize
\begin{semilogyaxis}[
legend columns=2,
legend cell align=left,
ymin=0.000001,
ymax=0.1,
width=3.2in,
height=2.1in,
grid=both,
xmin = 1,
xmax = 4,
xlabel = $E_b/N_0$ in dB,
ylabel = FER,
xtick={1,1.5,2,2.5,3,3.5,4,4.5,5,5.5}
]

\addplot[red, mark = o, only marks]
table[x=snr,y=fer]{snr fer
	0 0.4405
	0.25 0.3175
	0.5 0.246
	0.75 0.1655
	1 0.099000
	1.25 0.061000
	1.5 0.035000
	1.75 0.017933
	2 0.008000
	2.25 0.003292
	2.5 0.001176
	2.75 0.00043668
	3 0.000138
	3.25 3.8432e-05
	3.5 1.0009e-05
	3.75 3.0e-06
	4 7.9e-07
};\addlegendentry{\ac{SCOS}}

\addplot[brown, line width = 1.25pt]
table[x=snr,y=fer]{snr fer
	0 0.4405
	0.25 0.3175
	0.5 0.246
	0.75 0.1655
	1 0.099000
	1.25 0.061000
	1.5 0.035000
	1.75 0.017933
	2 0.008000
	2.25 0.003292
	2.5 0.001176
	2.75 0.00043668
	3 0.000138
	3.25 3.8432e-05
	3.5 1.0009e-05
	3.75 3.0e-06
	4 7.9e-07
};\addlegendentry{\ac{ML}}

\addplot[brown, dashed, line width=1.25pt]
table[x=snr,y=fer]{snr fer
1.000000000000000   0.143839366949861
1.250000000000000   0.082213310092274
1.500000000000000   0.042497569653307
1.750000000000000   0.019950901922751
2.000000000000000   0.008586425015708
2.250000000000000   0.003223271460151
2.500000000000000   0.001100579891236
2.750000000000000   0.000330079099603
3.000000000000000   0.000086745193797
3.250000000000000   0.000019995279927
3.500000000000000   0.000004066498212
3.750000000000000   0.000000708069023
};\addlegendentry{RCU}

\addplot[brown, dotted, line width=1.25pt]
table[x=snr,y=fer]{snr fer
1.000000000000000   0.084450067574090
1.250000000000000   0.046789395657259
1.500000000000000   0.023247677313516
1.750000000000000   0.010448670867371
2.000000000000000   0.004249278249449
2.250000000000000   0.001546310966604
2.500000000000000   0.000484410332805
2.750000000000000   0.000130900052036
3.000000000000000   0.000029657241148
3.250000000000000   0.000005796169358
3.500000000000000   0.000000972556776
};\addlegendentry{MC}

\end{semilogyaxis}
	\begin{axis}[/pgf/number format/.cd,precision=0,
width=3.2in,
height=2.1in,
xmin = 1,
xmax = 4,
        legend style={at={(0.86,0.725)},anchor=north,cells={align=left}},
		ylabel = {$\mathbb{E}\left[\Lambda\right]$},
		xticklabels=none,
		axis y line*=right,
		axis x line*=none,
		ymin=0, ymax=1e2,
		ytick = {1,20,40,60, 80, 100}
		]

\addplot[red, mark = *, dashed]
table[x=snr,y=fer]{snr fer
	0 460.1
	0.25 372.04
	0.5 311.22
	0.75 211.98
	1 156.06
	1.5 61.176
	2 20.436
	2.5 5.8956
	3 2.103
	3.5 1.264679
	4 1.062140
};\addlegendentry{ANV}

\addplot[black, dash dot, line width=1.25pt]
table[x=snr,y=fer]{snr fer
	0 170.43
	0.25 132.12
	0.5 101.3
	0.75 76.784
	1 50.058
	1.25 35.447
	1.5 21.566
	1.75 13.517
	2 7.6471
	2.25 4.4946
	2.5 2.9054
	2.75 1.9874
	3 1.5279
	3.25 1.2769
	3.5 1.149
	3.75 1.079853
	4.0 1.041982
};\addlegendentry{\eqref{eq:complexity_average}}

\end{axis}

\end{tikzpicture}
    \vspace{-5mm}
    \caption{FER/\ac{ANV} vs. $E_b/N_0$}
    \label{fig:perf}
\end{subfigure}
\begin{subfigure}{0.49\textwidth}
\hspace{3pt}
    	\begin{tikzpicture}[scale=1]
	\footnotesize
	\begin{axis}[
	ybar=1pt,
width=3.2in,
height=1.5in,
tickwidth=0pt,
ymajorgrids=true,
	bar width=0.12cm,
 symbolic x coords={0.5,1,1.5,2,2.5,3,3.5,4,4.5,5,5.5,6,6.5,7,7.5,8,8.5,9,9.5,$\geq$10,10.5},
	xtick=data,
	xlabel = $i$,
	ylabel = $\text{Pr}\left(i+1>\Lambda\geq i\right)$,
	ymin=0,
	ymax=1,
 tick align=inside,
 xtick style={/pgfplots/major tick length=0pt, },
    extra x ticks={0.5,...,10.5},
    extra x tick labels=\empty,
    extra x tick style={
    grid=major,
    xtick style={/pgfplots/major tick length=0pt,},
    },
    xmin = 0.5,
    xmax = 10.5,
	]
	\addplot[fill=brown,brown] table[x=e,y=n]{e n
1 0.45413
2 0.10969
3 0.066754
4 0.047909
5 0.035823
6 0.028121
7 0.022444
8 0.01922
9 0.015816
$\geq$10 0.2001
	};\addlegendentry{$E_b/N_0={2}~\text{dB}$}
	
	\addplot[fill=red!80!white,red!80!white] table[x=e,y=n]{e n
1 0.84158
2 0.065813
3 0.028324
4 0.015415
5 0.010114
6 0.0071199
7 0.0050587
8 0.0038656
9 0.0030208
$\geq$10 0.019686
	};\addlegendentry{$E_b/N_0={3}~\text{dB}$}
	
	\addplot[fill=blue,blue] table[x=e,y=n]{e n
1 0.98664
2 0.0085768
3 0.0023237
4 0.00097031
5 0.00049246
6 0.00028627
7 0.00018156
8 0.00012022
9 8.3132e-05
$\geq$10 0.00032138
	};\addlegendentry{$E_b/N_0={4}~\text{dB}$}

	\end{axis}
	\end{tikzpicture}
    \vspace{-5mm}
    \caption{\ed{Histogram vs. $E_b/N_0$}}
    \label{fig:hist}
\end{subfigure}
\caption{\edd{FER/\ac{ANV}\eddd{/histogram for node visits} vs. $E_b/N_0$ over the \ac{biAWGN} channel for the $\left(128,64\right)$ \ac{PAC} code under \ac{SCOS} decoding compared to relative \ac{RCU} bound/lower bound \eqref{eq:complexity_average}.}}
\label{fig:scos_pac}
\end{figure}
\begin{remark}
    Recall that the \ac{PM}~(\ref{eq:PM}) is calculated using the \ac{SC} decoding schedule, i.e., it ignores the frozen bits coming after the current decoding phase $i$. This means the size of the set~\eqref{eq:set_def} tends to be smaller for codes more suited for \ac{SC} decoding, e.g., polar codes, while it is larger for other codes such as \ac{RM} codes. This principle is also observed when decoding via \ac{SCL} decoding, i.e., the required list size to approach \ac{ML} performance grows when one ``interpolates'' from polar to \ac{RM} codes\cite{Mondelli14,RMpolar14,CP21}. This observation motivates \ed{us to introduce dRM-polar code ensemble in Definition~\ref{def:dRM}, whose random instances} provide a good performance vs. complexity trade-off under \ac{SCOS} decoding for moderate code lengths, e.g., $N=256$ bits.
\end{remark}

\subsection{Performance under Maximum Complexity Constraints}
\label{sec:comp_perf}

\edd{The proxy of node-visits is particularly useful when one wants to limit the worst-case complexity of \emph{polar code decoders} leveraging the structure of the Hadamard matrix in a unified manner. Observe that, for a given code construction \edd{specified by set $\mathcal{A}$}, the number $\lambda_{\text{SCL}}(L,\mathcal{A})$ of node-visits for \ac{SCL} decoding with list size $L$ is constant and upper bounded as $\lambda_{\text{SCL}}(L,\mathcal{A})\leq \edd{L}$. Given a polar code decoder, one may force it to satisfy a maximum number of node-visits such that $\lambda\left(y^N\right)\leq L$, for some positive integer $L$ at each decoding attempt with the hope that the worst-case complexity can be comparable to de facto reference \ac{SCL} decoding with list size $L$. For instance, sequential decoders use similar parameters, call $L$, to limit their worst-case complexity comparable to that of \ac{SCL} decoding with list size $L$\cite[Sec. III]{trifonov2018score}, \cite[Sec. V]{MMQ20}.} \eddd{Similarly,} \edd{\ac{SCOS} is modified by returning the existing most-likely candidate whenever a pre-defined maximum number of node-visits $\lambda_\text{max}$ is reached, at the expense of suboptimality. \edddd{Note that if $\lambda_\text{max}\geq 1$, SCOS decoding will always return a valid codeword}.}

\edd{Figures \ref{fig:scos_dfRMshort} and \ref{fig:scos_dfRMmoderate} provide performance and the \ac{ANV} vs. \ac{SNR} (in $E_b/N_0$) for short- and moderate-length \ac{dRM} and \ac{dRM}-polar codes, i.e., $N\in\{64,128\}$ and $N\in\{256,512\}$, respectively, of various rates ($0.14<R<0.92$) under \ac{SCOS} decoding, where $\lambda_\text{max}$ \eddd{and $\eta$ are both} set to $10$, $100$ and $5000$ for cases where $N=64$, $N=128$ and $N\in\{256,512\}$ for the simulations, respectively. For \ac{dRM} codes, the information sets are the same as the \ac{RM} code with the same block length and dimension, where the dynamic frozen bit constraints are randomly chosen. The $(256,154)$ dRM-polar code with a rate of $\approx 0.6$ is chosen uniformly at random from the ensemble of Definition \ref{def:dRM}, where the information set $\mathcal{A}$ is defined as in \cite{RMpolar14} with the mother RM$(4,8)$ code and the polar rule given by setting $\beta=2^{\nicefrac{1}{4}}$ in \cite{he2017beta}. Empirical \emph{ML lower bounds} of~\cite{tal2015list} are also plotted.} \edd{\ac{ML} decoding performance is \emph{approached} if the performance matches the simulated lower bounds.}
\begin{figure*}
\centering
\begin{subfigure}{0.49\textwidth}
    \begin{tikzpicture}[scale=1]
\footnotesize
\begin{semilogyaxis}[
legend style={at={(1,1)},anchor=north east},
legend columns=1,
legend cell align=left,
ymin=0.000001,
ymax=0.1,
width=3.25in,
height=2.1in,
grid=both,
xmin = 1.5,
xmax = 5,
xlabel = $E_b/N_0$ in dB,
ylabel = FER,
xtick={1,1.5,2,2.5,3,3.5,4,4.5,5}
]

\addplot[blue, mark = x]
table[x=snr,y=fer]{snr fer
1.00       0.083472  
1.50       0.054318    
2.00       0.018997    
2.50       0.008547    
3.00       0.002805    
3.50       0.000760    
4.00       0.000200   
4.50       0.000037
5.00       0.0000052
};\addlegendentry{\scriptsize{RM, ML}}

\addplot[red, mark = o, only marks]
table[x=snr,y=fer]{snr fer
1.00       0.0845959596    
1.50       0.0329081633    
2.00       0.0138020833    
2.50       0.00523316062    
3.00       0.0013486911    
3.50       0.000340512821    
4.00       0.0000898876404    
4.50       1.00597959e-5
5.00       1.30353673e-6
};\addlegendentry{\scriptsize{dRM, \ac{SCOS}}} 

\addplot[brown, line width = 1.25pt]
table[x=snr,y=fer]{snr fer
1.00       0.083750    
1.50       0.032250    
2.00       0.013250    
2.50       0.005050    
3.00       0.001288    
3.50       0.000332    
4.00       0.000080    
4.50       9.758002e-06
5.00       1.277466e-06
};\addlegendentry{\scriptsize{dRM, \ac{ML}}}

\addplot[brown, dashed, line width=1.25pt]
table[x=snr,y=fer]{snr fer
1.000000000000000   0.115886366213589
1.250000000000000   0.078306841497885
1.500000000000000   0.051426139481015
1.750000000000000   0.032269957928044
2.000000000000000   0.019472109272230
2.250000000000000   0.011188471710953
2.500000000000000   0.006113042960368
2.750000000000000   0.003173399363903
3.000000000000000   0.001565244822178
3.250000000000000   0.000734380896523
3.500000000000000   0.000325357361175
3.750000000000000   0.000138096889336
4.000000000000000   0.000056022452290
};
\addlegendentry{\scriptsize{RCU}}

\end{semilogyaxis}

	\begin{axis}[
width=3.25in,
height=2.1in,
		xmin = 1.5,
		xmax = 5,
  legend style={at={(1,0.62)},anchor=north east,cells={align=left}},
		ylabel = {$\mathbb{E}\left[\Lambda\right]$},
		xticklabels=none,
		axis y line*=right,
		axis x line*=none,
		ymin=1e0, ymax=2e0,
		ytick = {1,1.2,1.4,1.6,1.8,2}
		]

\addplot[red, dashed, mark = *, mark options = solid]
table[x=snr,y=fer]{snr fer
1.00     2.461620
1.50     1.926933
2.00     1.538337
2.50     1.275753
3.00     1.127301
3.50     1.052523
4.00     1.019462
4.50     1.006239
5.00     1.001712
};\addlegendentry{\scriptsize{dRM, ANV}} 

\end{axis}

\end{tikzpicture}
    \vspace{-5mm}
    \caption{$(64, 22)$}
    \label{fig:k22n64}
\end{subfigure}
\hfill
\begin{subfigure}{0.49\textwidth}
    \begin{tikzpicture}[scale=1]
\footnotesize
\begin{semilogyaxis}[
legend style={at={(1,1)},anchor=north east},
legend columns=1,
legend cell align=left,
ymin=0.000001,
ymax=0.1,
width=3.25in,
height=2.1in,
grid=both,
xmin = 2.5,
xmax = 5.5,
xlabel = $E_b/N_0$ in dB,
ylabel = FER,
xtick={1,1.5,2,2.5,3,3.5,4,4.5,5,5.5}
]

\addplot[blue, mark = x]
table[x=snr,y=fer]{snr fer
2.000000        0.124069
2.250000        0.074683
2.500000        0.051230
2.750000        0.033967
3.000000        0.016975
3.250000        0.009857
3.500000        0.005790
3.750000        0.003280
4.000000        0.001384
4.250000        0.000626
4.500000        0.000258
4.750000        0.000120
5.000000        0.000052
5.250000        0.000016
5.500000        0.00000627
5.750000        0.00000162
};\addlegendentry{\scriptsize{RM, ML}}
\addplot[red, mark = o, only marks]
table[x=snr,y=fer]{snr fer
1.00      0.31505102    
1.50      0.18877551    
2.00      0.0979591837    
2.50      0.0369897959   
3.00      0.0119525253    
3.50      0.00325520833    
4.00      0.000751041667    
4.50      0.00012371134    
5.00      1.7182134E-5
5.50      1.822290e-06
5.75      5.102457e-07
};\addlegendentry{\scriptsize{\ac{dRM}, \ac{SCOS}}} 

\addplot[brown, line width = 1.25pt]
table[x=snr,y=fer]{snr fer
1.00      0.308500    
1.25      0.233500    
1.50      0.185000    
1.75      0.132750    
2.00      0.096000    
2.25      0.054250    
2.50      0.036250    
2.75      0.022875    
3.00      0.011750    
3.25      0.005750    
3.50      0.003125    
3.75      0.001486    
4.00      0.000721    
4.25      0.000298    
4.50      0.000120    
4.75      4.355401e-05
5.00      1.666667e-05
5.25      5.299979e-06
5.50      1.822290e-06
5.75      5.102457e-07
};\addlegendentry{\scriptsize{\ac{dRM}, \ac{ML}}}

\addplot[brown, dashed, line width=1.25pt]
table[x=snr,y=fer]{snr fer
2.000000000000000   0.144639812349844
2.250000000000000   0.090642547980369
2.500000000000000   0.052949069328249
2.750000000000000   0.029791714731050
3.000000000000000   0.015903182473044
3.250000000000000   0.007915937963532
3.500000000000000   0.003682068501351
3.750000000000000   0.001644281764168
4.000000000000000   0.000698260209342
4.250000000000000   0.000280750003005
};
\addlegendentry{\scriptsize{RCU}}

\end{semilogyaxis}
	\begin{axis}[
width=3.25in,
height=2.1in,
xmin = 2.5,
xmax = 5.5,
  legend style={at={(1,0.62)},anchor=north east,cells={align=left}},
		ylabel = {$\mathbb{E}\left[\Lambda\right]$},
		xticklabels=none,
		axis y line*=right,
		axis x line*=none,
		ymin=1e0, ymax=2e0,
		ytick = {1,1.2,1.4,1.6,1.8,2}
		]

\addplot[red, mark = *, dashed]
table[x=snr,y=fer]{snr fer
1.00      3.468700
1.50      2.621128
2.00      1.970037
2.50      1.535428
3.00      1.230723
3.50      1.087558
4.00      1.030826
4.50      1.008787
5.00      1.002071
5.50      1.000413
5.75      1.000165
};\addlegendentry{\scriptsize{dRM, ANV}}

\end{axis}

\end{tikzpicture}
    \vspace{-5mm}
    \caption{$(64, 42)$}
    \label{fig:k42n64}
\end{subfigure}
\begin{subfigure}{0.49\textwidth}
    \begin{tikzpicture}[scale=1]
\footnotesize

\begin{semilogyaxis}[
legend style={at={(1,1)},anchor=north east},
legend cell align=left,
ymin=0.000001,
ymax=0.1,
width=3.25in,
height=2.1in,
grid=both,
xmin = 1,
xmax = 4,
xlabel = $E_b/N_0$ in dB,
ylabel = FER,
xtick={1,1.5,2,2.5,3,3.5,4,4.5,5,5.5}
]

\addplot[blue, mark = x]
table[x=snr,y=fer]{snr fer
0.000000        0.219298
0.250000        0.149477
0.500000        0.096432
0.750000        0.068353
1.000000        0.043234
1.250000        0.034566
1.500000        0.019755
1.750000        0.010401
2.000000        0.005155
2.250000        0.002992
2.500000        0.001832
2.750000        0.000853
3.000000        0.000320
3.250000        0.000150
3.500000        0.000058
3.750000        0.000022
4.000000        0.00000825
};\addlegendentry{\scriptsize{RM, ML}}

\addplot[red, mark = o, only marks]
table[x=snr,y=fer]{snr fer
1.00      0.0301630435    
1.50      0.0110526316     
2.00      0.00283163265     
2.50      0.000546     
3.00      0.000088    
3.50      9.28995811e-6
4.00      9.868551e-07
};\addlegendentry{\scriptsize{dRM, \ac{SCOS}}} 

\addplot[brown, line width = 1.25pt]
table[x=snr,y=fer]{snr fer
1.00      0.027750    
1.25      0.016375    
1.50      0.010500    
1.75      0.005150    
2.00      0.002775    
2.50      0.000546     
3.00      0.000088    
3.50      9.10415895e-6
3.75      3.006253e-06
4.00      9.868551e-07
};\addlegendentry{\scriptsize{dRM, \ac{ML}}} 

\addplot[brown, dashed, line width=1.25pt]
table[x=snr,y=fer]{snr fer
1.000000000000000   0.039152844678578
1.500000000000000   0.012228701907241
2.000000000000000   0.002977879167160
2.500000000000000   0.000533596547898
3.000000000000000   0.000069597916723
3.500000000000000   0.000006456518156
4.000000000000000   0.000000420982928
};\addlegendentry{\scriptsize{RCU}}
\end{semilogyaxis}

	\begin{axis}[
width=3.25in,
height=2.1in,
xmin = 1,
xmax = 4,
		ylabel = {$\mathbb{E}\left[\Lambda\right]$},
  legend style={at={(1,0.62)},anchor=north east,cells={align=left}},
		xticklabels=none,
		axis y line*=right,
		axis x line*=none,
		ymin=0, ymax=2e1,
		ytick = {1,4,8,12,16,20}
		]

\addplot[red, mark = *, dashed]
table[x=snr,y=fer]{snr fer
1.00      17.923721
1.50      9.555017
2.00      5.298481
2.50      3.025463 
3.00      1.852629 
3.50      1.342139
4.00      1.130070
};\addlegendentry{\scriptsize{dRM, ANV}}

\end{axis}

\end{tikzpicture}
    \vspace{-5mm}
    \caption{$(128, 29)$}
    \label{fig:k29n128}
\end{subfigure}
\hfill
\begin{subfigure}{0.49\textwidth}
    \begin{tikzpicture}[scale=1]
\footnotesize
\begin{semilogyaxis}[
legend style={at={(1,1)},anchor=north east},
legend columns=1,
legend cell align=left,
ymin=0.000001,
ymax=0.1,
width=3.25in,
height=2.1in,
grid=both,
xmin = 3,
xmax = 5.5,
xlabel = $E_b/N_0$ in dB,
ylabel = FER,
xtick={1,1.5,2,2.5,3,3.5,4,4.5,5,5.5}
]
\addplot[blue, mark = x]
table[x=snr,y=fer]{snr fer
1.000000        0.724638  
1.250000        0.609756  
1.500000        0.505051  
1.750000        0.418410  
2.000000        0.284091  
2.250000        0.210084  
2.500000        0.143266  
2.750000        0.071480  
3.000000        0.055804  
3.250000        0.026110  
3.500000        0.012243  
3.750000        0.005621  
4.000000        0.002578  
4.250000        0.000947  
4.500000        0.000377  
4.750000        0.000127  
5.000000        0.000037  
5.250000        0.0000116 
5.500000        0.00000293
};\addlegendentry{\scriptsize{RM, ML}}

\addplot[red, mark = o, only marks]
table[x=snr,y=fer]{snr fer 
2.25      0.150250     
2.50      0.096250     
3.00      0.029250     
3.50      0.005700    
4.00      0.000694    
4.50  	  6.561680e-05
5.00  	  4.162504e-06
5.50  	  3.104857e-07
};\addlegendentry{\scriptsize{dRM, \ac{SCOS}}} 

\addplot[brown, line width = 1.25pt]
table[x=snr,y=fer]{snr fer
2.25      0.150250     
2.50      0.096250     
3.00      0.029250     
3.50      0.005700    
4.00      0.000694    
4.50  	  6.561680e-05
5.00  	  4.162504e-06
5.50  	  3.104857e-07
};\addlegendentry{\scriptsize{dRM, \ac{ML}}}

\addplot[brown, dashed, line width=1.25pt]
table[x=snr,y=fer]{snr fer
2.500000000000000   0.152477613892065
2.750000000000000   0.078743804410255
3.000000000000000   0.037115035717574
3.250000000000000   0.015993086140909
3.500000000000000   0.005991806919906
3.750000000000000   0.002062012532753
4.000000000000000   0.000620769214744
4.250000000000000   0.000176584839303
4.500000000000000   0.000042820344892
4.750000000000000   0.000010636837909
};
\addlegendentry{\scriptsize{RCU}}

\end{semilogyaxis}
	\begin{axis}[
width=3.25in,
height=2.1in,
xmin = 3,
xmax = 5.5,
legend style={at={(1,0.62)},anchor=north east,cells={align=left}},
		ylabel = {$\mathbb{E}\left[\Lambda\right]$},
		xticklabels=none,
		axis y line*=right,
		axis x line*=none,
		ymin=1, ymax=4e0,
		ytick = {1,1.6,2.2,2.8,3.4,4}
		]

\addplot[red, mark = *, dashed]
table[x=snr,y=fer]{snr fer
2.50      8.082700  
3.00      3.569640  
3.50      1.745962 
4.00      1.190449 
4.50  	  1.042782
5.00  	  1.008514
5.50  	  1.001391
};\addlegendentry{\scriptsize{dRM, ANV}}

\end{axis}
\end{tikzpicture}
    \vspace{-5mm}
    \caption{$(128, 99)$}
    \label{fig:k99n128}
\end{subfigure}
\caption{\edd{FER/\ac{ANV} vs. $E_b/N_0$ over the \ac{biAWGN} channel for short \ac{dRM} codes under \ac{SCOS} decoding with $\lambda_\text{max}\eddd{\,=\eta }=10$ and $\lambda_\text{max}\eddd{\,=\eta}=100$ for $N=64$ and $N=128$, respectively, compared to relative RM codes under ML decoding and RCU bounds.}}
\label{fig:scos_dfRMshort}
\end{figure*}
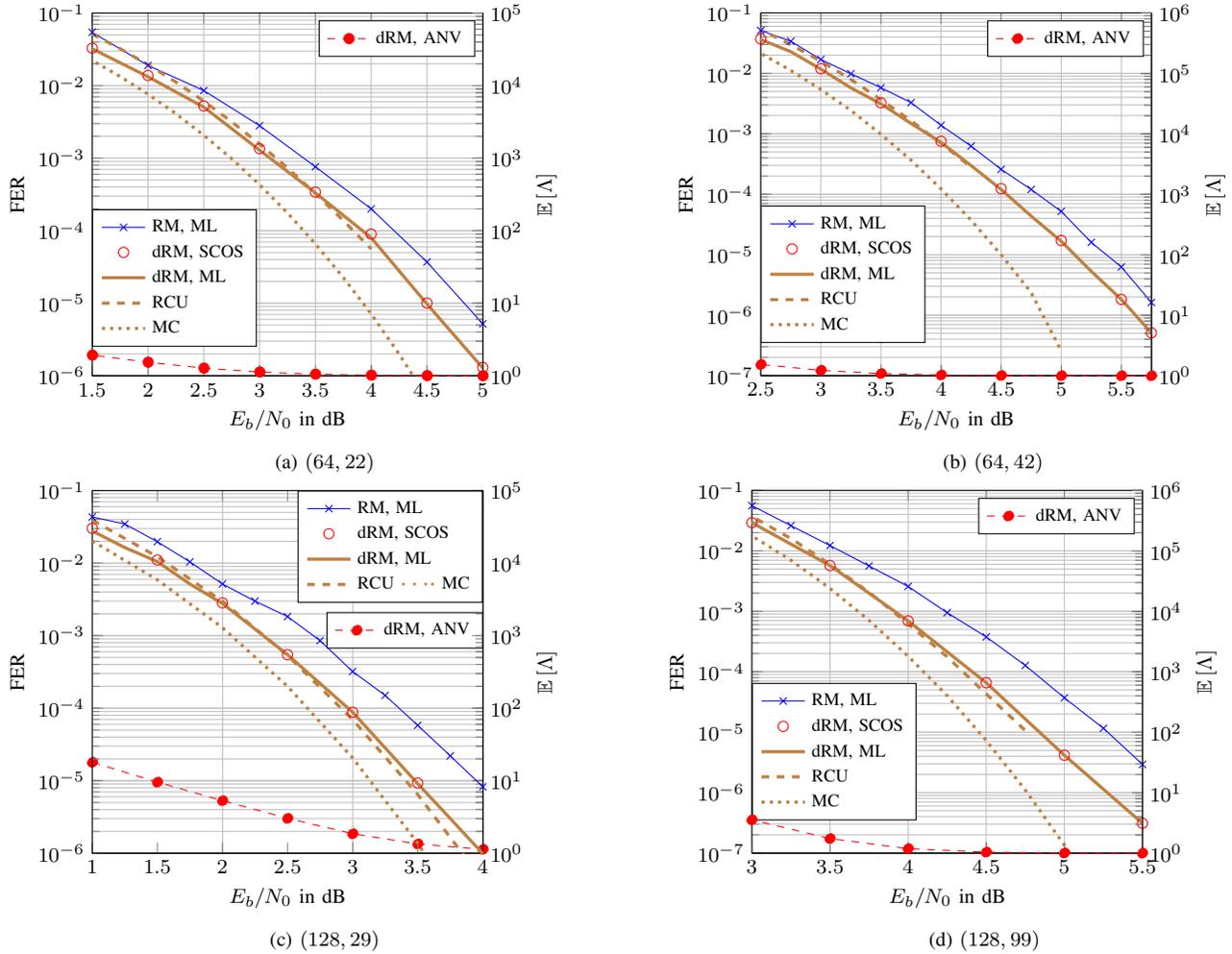
\begin{figure*}
\centering
\begin{subfigure}{0.49\textwidth}
    \begin{tikzpicture}[scale=1]
\footnotesize
\begin{semilogyaxis}[
legend pos = south west,
legend columns=1,
legend cell align=left,
ymin=0.0000001,
ymax=0.1,
width=3.25in,
height=2.6in,
grid=both,
xmin = 0.5,
xmax = 3.25,
xlabel = $E_b/N_0$ in dB,
ylabel = FER,
xtick={0,0.5,1,1.5,2,2.5,3,3.5,4,4.5,5,5.5}
]

\addplot[blue, mark = x]
table[x=snr,y=fer]{snr fer
0.000000        0.105042        0.105042
0.250000        0.082919        0.082919
0.500000        0.053362        0.053362
0.750000        0.028450        0.028450
1.000000        0.021026        0.021026
1.250000        0.008713        0.008713
1.500000        0.005912        0.005912
1.750000        0.002142        0.002142
2.000000        0.001168        0.001168
2.250000        0.000518        0.000518
2.500000        0.000193        0.000193
2.750000        0.000063        0.000063
3.000000        0.000026        0.000026
3.250000        0.000011        0.000011
};\addlegendentry{\scriptsize{RM, ML}}

\addplot[red, mark = o, only marks]
table[x=snr,y=fer]{snr fer
0.000000        0.090827
0.250000        0.053821
0.500000        0.035361
0.750000        0.020076
1.00    0.009446  
1.25    0.004708  
1.50    0.002273  
1.75    0.000929  
2.00    0.000354  
2.25    0.000134  
2.50    0.00004288
2.75    0.00001063
3.00    0.00000241
3.25    0.00000055
};\addlegendentry{\scriptsize{dRM, SCOS}}

\addplot[brown, line width = 1.25pt]
table[x=snr,y=fer]{snr fer
0.000000        0.090827
0.250000        0.053821
0.500000        0.035361
0.750000        0.020076
1.00    0.009446  
1.25    0.004667  
1.50    0.002227  
1.75    0.000875  
2.00    0.000337  
2.25    0.000131  
2.50    0.00004288
2.75    0.00001042
3.00    0.00000236
3.25    0.00000055
};\addlegendentry{\scriptsize{dRM, \ac{ML}}}

\addplot[brown, dashed, line width=1.25pt]
table[x=snr,y=fer]{snr fer
0.000000000000000   0.119998128089523
0.250000000000000   0.074639000838171
0.500000000000000   0.043916599173199
0.750000000000000   0.024031092477856
1.000000000000000   0.012319325413008
1.500000000000000   0.002540258823345
2.000000000000000   0.000363271361524
2.500000000000000   0.000034162681044
3.000000000000000   0.000002007804938
};
\addlegendentry{\scriptsize{RCU}}

\addplot[brown, dotted, line width=1.25pt]
table[x=snr,y=fer]{snr fer
0.000000000000000   0.072500305774611
0.250000000000000   0.044786249195379
0.500000000000000   0.025905885946289
0.750000000000000   0.013892880023224
1.000000000000000   0.006787981604876
1.250000000000000   0.003150247958208
1.500000000000000   0.001318142775469 
1.750000000000000   0.000492053012025
2.000000000000000   0.000167798618724
2.500000000000000   0.000013423720919
3.000000000000000   0.000000621987691
3.500000000000000   0.000000016051681
};\addlegendentry{\scriptsize{MC}}

\end{semilogyaxis}
	\begin{semilogyaxis}[
width=3.25in,
height=2.6in,
xmin = 0.5,
xmax = 3.25,
		ylabel = {$\mathbb{E}\left[\Lambda\right]$},
		xticklabels=none,
		axis y line*=right,
		axis x line*=none,
		ymin=1e0, ymax=1e3,
		]

\addplot[red, mark = *, dashed]
table[x=snr,y=fer]{snr fer
0.000000       829.421038
0.250000    743.953567
0.500000    724.071748
0.750000    676.761585
1.00    622.425858
1.25    587.860770
1.50    550.256064
1.75    476.064335
2.00    339.352031
2.25    185.839074
2.50    80.409015
2.75    34.919000
3.00    18.724190
3.25    11.570512
};\addlegendentry{\scriptsize{dRM, ANV}}

\end{semilogyaxis}
\end{tikzpicture}
    \vspace{-5mm}
    \caption{$(256, 37)$ dRM code.}
\end{subfigure}
\hfill
\begin{subfigure}{0.49\textwidth}
    \begin{tikzpicture}[scale=1]
\footnotesize
\begin{semilogyaxis}[
legend pos = south west,
legend cell align=left,
ymin=0.0000001,
ymax=0.1,
width=3.25in,
height=2.6in,
grid=both,
xmin = 1.5,
xmax = 3.5,
xlabel = $E_b/N_0$ in dB,
ylabel = FER,
xtick={1,1.5,2,2.5,3,3.5,4,4.5,5,5.5}
]

\addplot[blue, mark = x]
table[x=snr,y=fer]{snr fer
1.000000        0.279202
1.500000        0.080386
2.000000        0.012806
2.500000        0.001983
3.000000        0.000168
3.500000        0.000016
};\addlegendentry{\scriptsize{RM-p., ML}}

\addplot[red, mark = o, only marks]
table[x=snr,y=fer]{snr fer
1.000000        0.250627     
1.250000        0.136799     
1.500000        0.047642     
1.750000        0.022085     
2.000000        0.007084     
2.250000        0.001787     
2.500000        0.000552     
2.750000        0.000093     
3.000000        0.000020     
3.250000        0.00000408   
3.500000        0.00000056889
};\addlegendentry{\scriptsize{dRM-p., \ac{SCOS}}}

\addplot[brown, line width = 1.25pt]
table[x=snr,y=fer]{snr fer
1.00   0.210526
1.25   0.127223
1.50   0.045260
1.75   0.020981
2.00   0.006591
2.25   0.001655
2.50   0.000381
2.75   0.000070
3.00   0.00001338
3.25   0.00000355
3.50   0.000000534759
};\addlegendentry{\scriptsize{dRM-p., \ac{ML}}}

\addplot[brown, dashed, line width=1.25pt]
table[x=snr,y=fer]{snr fer
1.000000   0.281811597217360
1.500000   0.059757617047152
1.750000   0.021787891578998
2.000000   0.006164860529908
2.250000   0.001447497327921
2.500000   0.000267034704852
2.750000   0.000038281156687
3.000000   0.000004237053290
3.250000   0.000000361836078
3.500000   0.000000022242626
};\addlegendentry{\scriptsize{RCU} \textcolor{brown}{$\boldsymbol{\cdot\cdot\cdot\cdot}$} \scriptsize{MC}}

\addplot[brown, dotted, line width=1.25pt]
table[x=snr,y=fer]{snr fer
1.000000000000000   0.189108162000019
1.250000000000000   0.091949523847831
1.500000000000000   0.038472876420095
1.750000000000000   0.013630064420632
2.000000000000000   0.003764763316368
2.250000000000000   0.000784435440668
2.500000000000000   0.000146821299603
2.750000000000000   0.000017424951529
3.000000000000000   0.000001717851313
3.250000000000000   0.000000114536883
3.500000000000000   0.000000006313519
};

\end{semilogyaxis}
	\begin{semilogyaxis}[
width=3.25in,
height=2.6in,
xmin = 1.5,
xmax = 3.5,
		ylabel = {$\mathbb{E}\left[\Lambda\right]$},
		xticklabels=none,
		axis y line*=right,
		axis x line*=none,
		ymin=1e0, ymax=1e3,
		]

\addplot[red, mark = *, dashed]
table[x=snr,y=fer]{snr fer
1.250000        1507.919529
1.500000        799.158483 
1.750000        415.503528 
2.000000        196.295577 
2.250000        71.604317  
2.500000        24.860250  
2.750000        7.850612   
3.000000        2.842904   
3.250000        1.477260   
3.500000        1.128762
};\addlegendentry{\scriptsize{dRM-p., ANV}}

\end{semilogyaxis}

\end{tikzpicture}
    \vspace{-5mm}
    \caption{$(256, 154)$ dRMpolar code.}
\end{subfigure}
\begin{subfigure}{0.49\textwidth}
    \begin{tikzpicture}[scale=1]
\footnotesize
	\begin{semilogyaxis}[
width=3.25in,
height=2.6in,
		grid=both,
xmin = 3.5,
xmax = 5.5,
        legend style={at={(0.818,0.65)},anchor=north,cells={align=left}},
		ylabel = {$\mathbb{E}\left[\Lambda\right]$},
		xticklabels=none,
		axis y line*=right,
		axis x line*=none,
		ymin=1e0, ymax=1e3,
		]

\addplot[red, mark = *, dashed]
table[x=snr,y=fer]{snr fer
3.00    115.279784
3.25    61.297753
3.50    25.046607
3.75    11.075676
4.00    4.755738
4.25    2.337769
4.50 	1.476555
4.75 	1.174419  
5.00 	1.066131 
5.25 	1.025320
5.50    1.015320
};\addlegendentry{\scriptsize{dRM, ANV}}
\end{semilogyaxis}
\begin{semilogyaxis}[
legend columns=1,
legend cell align=left,
ymin=0.0000001,
ymax=0.1,
width=3.25in,
height=2.6in,
xmin = 3.5,
xmax = 5.5,
xlabel = $E_b/N_0$ in dB,
ylabel = FER,
xtick={1,1.5,2,2.5,3,3.5,4,4.5,5,5.5}
]

\addplot[blue, mark = x]
table[x=snr,y=fer]{snr fer
3.000000        0.194175
3.250000        0.089686
3.500000        0.040568
3.750000        0.018399
4.000000        0.006652
4.250000        0.002586
4.500000        0.000876
4.750000        0.000243
5.000000        0.000071
5.250000        0.000021
5.500000        0.000005
};\addlegendentry{\scriptsize{RM, ML}}

\addplot[red, mark = o, only marks]
table[x=snr,y=fer]{snr fer 
 3.000000        0.129702
 3.250000        0.065274
 3.500000        0.028043
 3.750000        0.007211
 4.000000        0.002282
 4.250000        0.000693
4.500000        0.000160
4.750000        0.000028
5.000000        0.00000456
5.250000        0.000000897
5.500000        0.00000023
};\addlegendentry{\scriptsize{dRM, \ac{SCOS}}} 

\addplot[brown, line width = 1.25pt]
table[x=snr,y=fer]{snr fer
3.000000        0.129702
 3.250000        0.065274
 3.500000        0.028043
 3.750000        0.007211
 4.000000        0.002282
 4.250000        0.000693
4.500000        0.000160
4.750000        0.000028
5.000000        0.00000456
5.250000        0.000000897
5.500000        0.00000023
};\addlegendentry{\scriptsize{dRM, \ac{ML}}}

\addplot[brown, dashed, line width=1.25pt]
table[x=snr,y=fer]{snr fer
3.000000        0.211070036594603
3.250000        0.094692211832993
3.500000        0.036898587538416
3.750000        0.011783037027200
4.000000        0.003050180352426
4.250000        0.000637917301447
4.500000        0.000118634699353
4.750000        0.000016996107601
5.000000        0.000002149996985
5.200000        4.067531259731501e-07
};
\addlegendentry{\scriptsize{RCU}}

\end{semilogyaxis}

\end{tikzpicture}
    \vspace{-5mm}
    \caption{$(256, 219)$ dRM code.}
\end{subfigure}
\hfill
\begin{subfigure}{0.49\textwidth}
    \begin{tikzpicture}[scale=1]
\footnotesize
	\begin{semilogyaxis}[
		legend pos = north east,
width=3.25in,
height=2.6in,
xmin = 4.0,
xmax = 5.5,
		ylabel = {$\mathbb{E}\left[\Lambda\right]$},
        grid=both,
		xticklabels=none,
		axis y line*=right,
		axis x line*=none,
		ymin=1e0, ymax=1e3,
		]

\addplot[red, mark = *, dashed]
table[x=snr,y=fer]{snr fer
4.00 	375.127783
4.25 	126.524038 
4.50 	33.133494
4.75 	8.583325  
5.00 	2.593536  
5.25 	1.387465  
5.50 	1.116407
};\addlegendentry{\scriptsize{dRM, ANV}}
\end{semilogyaxis}
\begin{semilogyaxis}[
legend pos = south west,
legend columns=1,
legend cell align=left,
ymin=0.0000001,
ymax=0.1,
width=3.25in,
height=2.6in,
xmin = 4.0,
xmax = 5.5,
xlabel = $E_b/N_0$ in dB,
ylabel = FER,
xtick={1,1.5,2,2.5,3,3.5,4,4.5,5,5.5}
]

\addplot[blue, mark = x]
table[x=snr,y=fer]{snr fer
4.000000        0.052690
4.500000        0.006101
5.000000        0.000304
5.500000        0.000017
};\addlegendentry{\scriptsize{RM, ML}}

\addplot[red, mark = o, only marks]
table[x=snr,y=fer]{snr fer 
4.00 	0.035625   
4.25 	0.009956   
4.50 	0.001950  
4.75 	0.000325  
5.00 	4.47771486e-5
5.25 	4.50127008e-6
5.50 	4.42086358e-7
};\addlegendentry{\scriptsize{dRM, \ac{SCOS}}} 

\addplot[brown, line width = 1.25pt]
table[x=snr,y=fer]{snr fer
4.00 	0.033125   
4.25 	0.007639   
4.50 	0.001635  
4.75 	0.000242  
5.00 	0.00002305
5.25 	0.00000264
5.50 	0.00000033
};\addlegendentry{\scriptsize{dRM, \ac{ML}}}

\addplot[brown, dashed, line width=1.25pt]
table[x=snr,y=fer]{snr fer
4.000000000000000   0.045240299409712
4.250000000000000   0.011235850150427
4.500000000000000   0.001912491564970
4.750000000000000   0.000264390836575
5.000000000000000   0.000024939461958
5.250000000000000   0.000001424893091
5.500000000000000   0.000000088636677
};
\addlegendentry{\scriptsize{RCU}}

\end{semilogyaxis}

\end{tikzpicture}
    \vspace{-5mm}
    \caption{$(512, 466)$ dRM code.}
\end{subfigure}
\caption{\edd{FER/\ac{ANV} vs. $E_b/N_0$ over the \ac{biAWGN} channel for moderate-length dRM and dRM-polar codes under \ac{SCOS} decoding $\lambda_\text{max}\eddd{\,=\eta}=5000$ compared to relative RM and RM-polar codes and RCU bounds.}}
\label{fig:scos_dfRMmoderate}
\end{figure*}
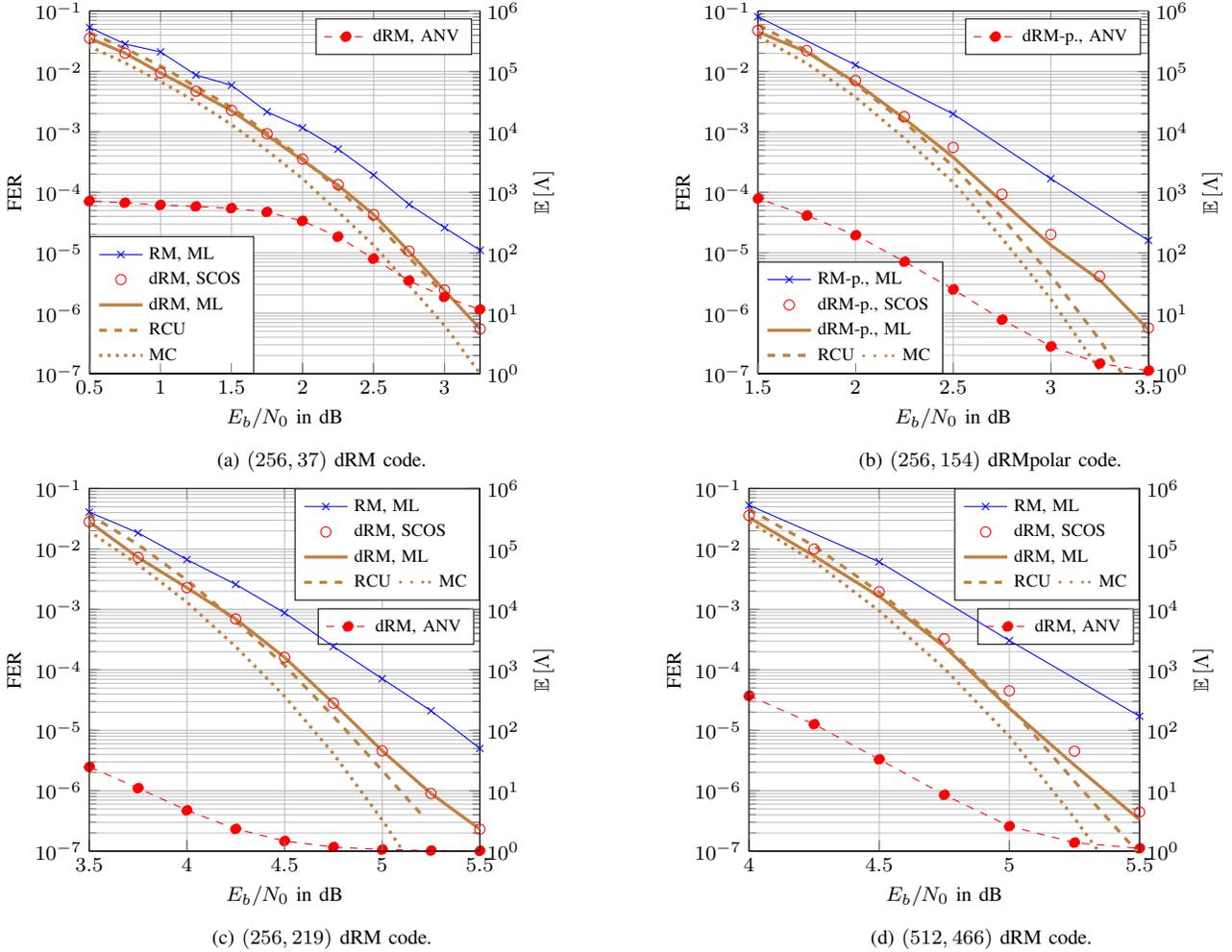

\edd{For all simulated codes, the performance approaches to that of \ac{ML} decoding (if not the same), which outperforms the \ac{ML} performance of \ac{RM} codes of the same parameters by up to $0.5$ dB, which is provided as reference. Note that \ac{SCOS} decoding parameters are set to the same complexity constraints for simulating \ac{RM} codes, where the \acp{FER} match the \ac{ML} lower bounds. Observe also that this remarkable performance is attained with $\mathbb{E}\left[\Lambda\right]\approx 1$ at high \ac{SNR} regime except for the case of $(256,37)$ codes, i.e., the proxy implies that the complexity will be close to that of \ac{SC} decoding at low \acp{FER} (e.g., $10^{-5}$ and below).}

Since the number of node-visits may not refer to the exact complexity, Tables \ref{tab:RM} and \ref{tab:dRM} provide the respective average complexity scores for \ac{SCOS}\eddd{, where $\lambda_{\text{max}}$ and $\eta$ are the same as in Figs. \ref{fig:scos_dfRMshort} and \ref{fig:scos_dfRMmoderate},} as well as \ac{SC} decoding of the codes as follows. During our simulations, we count the number of arithmetic operations, which include the floating-point additions and comparisons as well as binary XORs. In order to provide a unified complexity score, we assume that $1$ floating-point addition corresponds roughly to $8$ binary operations and $1$ floating-point comparison correspond to $6$ binary operations, i.e., an average complexity score is computed as
\begin{equation}
    8\cdot A + 6\cdot C + B\label{eq:complexityscore}
\end{equation}
where $A$, $B$ and $C$ are the average number of additions, XORs and comparisons during the simulations. These factors are motivated by the potential use of $8$-bit representation of real numbers in order to limit quantization errors and the fact that the comparison of two real numbers may be terminated before comparing all $8$ bits. Note that considered decoders do not require multiplications as the min-sum approximation is adopted. For instance, \ac{SC} decoding requires the same number of $\frac{1}{2}N\log_2 N$ additions, comparisons as well as XORs for any code of length $N$, resulting in a fixed complexity score as
\begin{equation}
    15\times\frac{1}{2}N\log_2 N \label{eq:sc_comp}
\end{equation}
if all frozen bits are set to $0$. Except for the case of $(256,37)$ codes, the average complexity score of \ac{SCOS} decoding is within a factor of $\edddd{1.5}$ from \ac{SC} decoding  for all \ac{RM} and \ac{dRM} codes at given \ac{SNR} values provided in the tables.
\setlength{\tabcolsep}{4pt}
\begin{table}
	\caption{\edd{RM codes under SCOS decoding}}
        \vspace{-3mm}
	\begin{center}
	\begin{tabular}{ccccccc}
		\hline\hline
		$(N,K)$ & $E_b/N_0$ & Addition & Compare & XOR & \eqref{eq:complexityscore} & \eddd{SC} \eqref{eq:sc_comp} \\
		\hline
		$(64,22)$ & $5.00$ & $300$ & $\edddd{192}$ & $193$ & $\edddd{3745}$ & $2880$ \\
		$(64,42)$ & $5.75$ & $339$ & $\edddd{192}$ & $193$ & $\edddd{4057}$ & $2880$ \\
		$(128,29)$ & $4.00$ & $724$ & $\edddd{504}$ & $507$ & $\edddd{9323}$ & $6720$ \\
		$(128,99)$ & $5.50$ & $774$ & $\edddd{448}$ & $450$ & $\edddd{9330}$ & $6720$\\
		$(256,37)$ & $3.25$ & $17943$ & $\edddd{12367}$ & $12204$ & $\edddd{229950}$ & $15360$\\
        $(256,219)$ & $5.50$ & $1733$ & $\edddd{1032}$ & $1034$ & $\edddd{21090}$ & $15360$ \\
        $(512,466)$ & $5.50$ & $4121$ & $\edddd{2515}$ & $2521$ & $\edddd{50579}$ & $34560$ \\
		\hline\hline
	\end{tabular}
\end{center}
	\label{tab:RM}
\end{table}
\begin{table}
	\caption{\edd{dRM codes under SCOS decoding}}
        \vspace{-3mm}
	\begin{center}
	\begin{tabular}{ccccccc}
		\hline\hline
		$(N,K)$ & $E_b/N_0$ & Addition & Compare & XOR & \eqref{eq:complexityscore} & SC \eqref{eq:complexityscore} \\
		\hline
		$(64,22)$ & $5.00$ & $300$ & $\edddd{192}$ & $269$ & $\edddd{3821}$ & $2948$ \\
		$(64,42)$ & $5.50$ & $339$ & $\edddd{192}$ & $266$ & $\edddd{4130}$ & $2954$ \\
		$(128,29)$ & $4.00$ & $724$ & $\edddd{505}$ & $763$ & $\edddd{9585}$ & $6956$ \\
		$(128,99)$ & $5.50$ & $774$ & $\edddd{449}$ & $708$ & $\edddd{9594}$ & $6945$ \\
		$(256,37)$ & $3.25$ & $18309$ & $\edddd{12793}$ & $19418$ & $\edddd{242648}$ & $16036$ \\
        $(256,219)$ & $5.50$ & $1733$ & $\edddd{1033}$ & $1697$ & $\edddd{21759}$ & $16027$ \\
        $(512,466)$ & $5.50$ & $4155$ & $\edddd{2562}$ & $4553$ & $\edddd{53165}$ & $36338$ \\
		\hline\hline
	\end{tabular}
\end{center}
	\label{tab:dRM}
\end{table}

\section{Further Improvements}
\label{sec:improvements}
\edd{An interesting modification to \ac{SCOS} is proposed by \cite{hashemi2021tree}, which avoids using \ac{PM} or any reliability score for the search as follows.} After finishing an instance of the initial \ac{SC} decoding, the flipping sets are prioritized according to a predefined depth-first or breadth-first search order. Numerical results show that the latter required less number of \ac{SC} decoding attempts for various \ac{RM} codes of length up to $512$ bits for approaching their \ac{ML} decoding performance compared to \ac{SCOS} decoding although each attempt is expected to require more node-visits. \edd{Nevertheless, the results imply that SCOS provides robust performance with various search schedules. In the following, we investigate the effect of search schedule via simulations when the bit reliability under \ac{SC} decoding is ignored, i.e., when only \acp{PM} are used for the search.}
\subsection{Bias Term Robustness}
\label{sec:bias}
Consider the bias terms $b_i$ given in \eqref{eq:score_bias_term}, \mcc{which} impacts the search priority \mcc{but not the performance \edd{if the maximum complexity constraints are unbounded}. This means that} a \emph{suboptimal} bias term does not change the performance of \ac{SCOS} decoding with unbounded complexity \mcc{(which is still \ac{ML} \edd{decoding})}, but \eddd{it may increase the} complexity. 

\eddd{In order to compute $b_i$, we assume that the all-zero codeword is transmitted thanks to the channel symmetry and the linearity of the codes under consideration. Let $f^{(i)}_{N}$ denote the \ac{PDF} of the \ac{RV} corresponding to $\ell_i(\boldsymbol{0})$, where $\boldsymbol{0}$ denotes an all-zero vector of length $i$ and the source of randomness is the channel output $Y^N$. Over general \acp{B-DMC}, the densities can be computed recursively as
\begin{align}
	f^{(2i-1)}_{N} = f^{(i)}_{\nicefrac{N}{2}} \boxast f^{(i)}_{\nicefrac{N}{2}}\label{eq:de_awgn1} \\
    f^{(2i)}_{N} = f^{(i)}_{\nicefrac{N}{2}} \varoast f^{(i)}_{\nicefrac{N}{2}}
	\label{eq:de_awgn2}
\end{align}
where $f^{(1)}_{1}$ is the \ac{PDF} of the \ac{i.i.d.} \acp{LLR} at the channel output, and $\boxast$ and $\varoast$ denote the check and variable node convolutions, respectively, as defined in \cite[Ch. 4]{Richardson:2008:MCT:1795974}. Then, terms $p_j$ in Eq. \eqref{eq:bias_term} can be computed via $f^{(i)}_{N}$ as
\begin{equation}
    p_j=\lim\limits_{z\rightarrow 0} \left(\int_{-\infty}^{-z}f^{(j)}_{N}(x)dx+\frac{1}{2}\int_{-z}^{+z}f^{(j)}_{N}(x)dx\right).\label{eq:de_awgn3}
\end{equation}
The computation of \eqref{eq:de_awgn1}, \eqref{eq:de_awgn2} and \eqref{eq:de_awgn3} can be carried out, for instance, via quantized density evolution \cite{chung_urbanke2001}, yielding an accurate estimate of the \ac{RHS} of \eqref{eq:de_awgn3}.}

Figure~\ref{fig:scos_dfpolar_robust} illustrates the effect of various bias terms outlined below on the performance of \ac{SCOS} decoding \edd{with bounded and unbounded $\lambda_\text{max}$.}
\begin{itemize}
	\item The bias terms are computed via the \ac{RHS} of~\eqref{eq:bias_term} using \eddd{quantized} density evolution for each \ac{SNR} point.
	\item \mcc{The bias terms are set to zero, i.e., $b_i=0$, $i\in[N]$ \eddd{which results in using the \ac{PM} as score as well}.}
\end{itemize}
The complexity reduction is limited if \eqref{eq:bias_term} is used instead of setting the bias terms to zero. Nevertheless, setting them to zero slightly degrades the performance (by $\approx0.12$~dB) when the maximum \edd{number of node-visits} is constraint to five times that of \ac{SC} decoding with almost no savings in the average complexity. Hence, we conclude that \ac{SCOS} decoding is not very sensitive to the choice of bias terms \edd{unless the maximum complexity is required to be very low}.
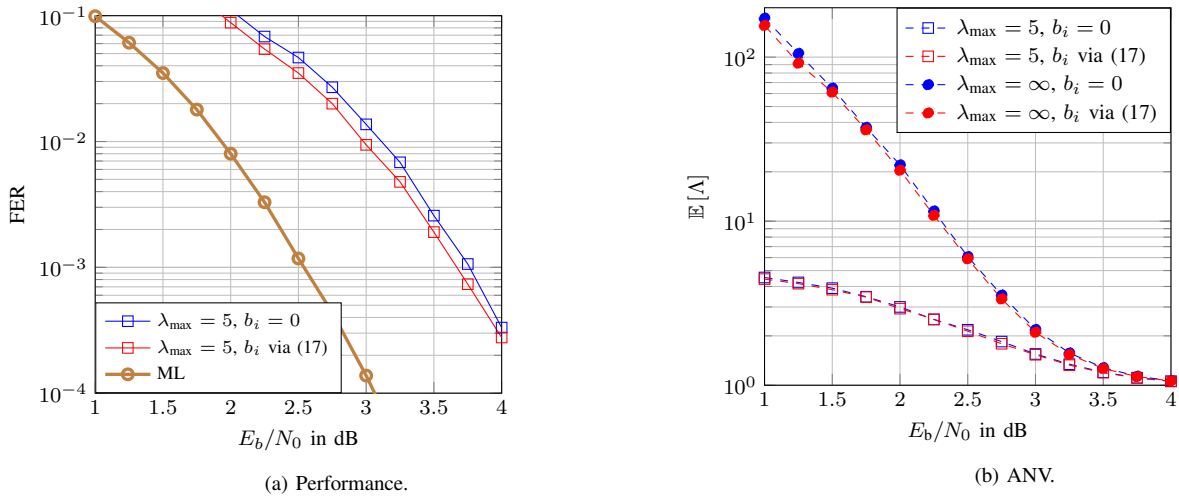
\begin{figure*}
\centering
\begin{subfigure}{0.49\textwidth}
    \begin{tikzpicture}[scale=1]
\footnotesize
\begin{semilogyaxis}[
legend pos=south west,
legend columns=1,
legend cell align=left,
ymin=0.0001,
ymax=0.1,
width=3.5in,
height=2in,
grid=both,
xmin = 1.0,
xmax = 4.0,
xlabel = $E_b/N_0$ in dB,
ylabel = FER,
xtick={1,1.5,2,2.5,3,3.5,4,4.5,5,5.5}
]

\addplot[blue, mark = square]
table[x=snr,y=fer]{snr fer 
1 0.419
1.25 0.3275
1.5 0.249
1.75 0.1725
2 0.111
2.25 0.06825
2.5 0.0465
2.75 0.027
3 0.013688
3.25 0.0068333
3.5 0.0025769
3.75 0.0010628
4 0.00033223
};\addlegendentry{$b_i=0$} 

\addplot[red, mark = square]
table[x=snr,y=fer]{snr fer 
1 0.412
1.25 0.3305
1.5 0.229
1.75 0.144
2 0.088
2.25 0.05425
2.5 0.035
2.75 0.02
3 0.0094091
3.25 0.0047778
3.5 0.0019057
3.75 0.00073529
4 0.00027701
};\addlegendentry{$b_i$ via \eqref{eq:bias_term}} 

\addplot[brown, mark=o, line width = 1.25pt]
table[x=snr,y=fer]{snr fer
	0 0.4405
	0.25 0.3175
	0.5 0.246
	0.75 0.1655
	1 0.099000
	1.25 0.061000
	1.5 0.035000
	1.75 0.017933
	2 0.008000
	2.25 0.003292
	2.5 0.001176
	2.75 0.00043668
	3 0.000138
	3.25 3.8432e-05
	3.5 1.0009e-05
	3.75 3.0e-06
	4 7.9e-07
};\addlegendentry{\ac{ML}}

\end{semilogyaxis}

\end{tikzpicture}
    \vspace{-2mm}
    \caption{Performance \eddd{with $\lambda_\text{max}=5$}.}
\end{subfigure}
\vspace{2mm}
\begin{subfigure}{0.49\textwidth}
    \begin{tikzpicture}[scale=1]
\footnotesize
\begin{semilogyaxis}[
width=3.5in,
height=2in,
legend cell align=left,
ymin=1,
ymax=200,
grid=both,
xmin = 1,
xmax = 4,
xlabel = $E_\text{b}/N_0\ \text{in dB}$,
ylabel = {$\mathbb{E}\left[\Lambda\right]$},
xtick={1,1.5,2,2.5,3,3.5,4}
]
\addplot[blue, mark = square, dashed, mark options = solid]
table[x=snr,y=fer]{snr fer
1 4.5375
1.25 4.2342
1.5 3.9061
1.75 3.4649
2 3.0024
2.25 2.516
2.5 2.1786
2.75 1.8461
3 1.5511
3.25 1.344
3.5 1.2011
3.75 1.1121
4 1.0578
};\addlegendentry{$\lambda_\text{max}=5$, $b_i=0$}
\addplot[red, mark = square, dashed, mark options = solid]
table[x=snr,y=fer]{snr fer
1 4.4367
1.25 4.1614
1.5 3.8236
1.75 3.4435
2 2.9312
2.25 2.5255
2.5 2.1347
2.75 1.7909
3 1.5346
3.25 1.3259
3.5 1.1927
3.75 1.1065
4 1.0555
};\addlegendentry{$\lambda_\text{max}=5$, $b_i$ via \eqref{eq:bias_term}}

\addplot[blue, mark = *, dashed] 
table[x=snr,y=fer]{snr fer
1.000000      172.79056
1.250000      105.35080
1.500000      64.691281
1.750000      37.307485
2.000000      22.029296
2.250000      11.541460
2.500000      6.082169 
2.750000      3.546714 
3.000000      2.186662 
3.250000      1.573250 
3.500000      1.279914 
3.750000      1.135098 
4.000000      1.064827 
};\addlegendentry{$\lambda_\text{max}=\infty$, $b_i=0$}

\addplot[red, mark = *, dashed]
table[x=snr,y=fer]{snr fer
	0 460.1
	0.25 372.04
	0.5 311.22
	0.75 211.98
	1 156.06
	1.25 91.848
	1.5 61.176
	1.75 36.022
	2 20.436
	2.25 10.851
	2.5 5.8956
	2.75 3.358
	3 2.103
	3.25 1.540764
	3.5 1.264679
	3.75 1.128713
	4 1.062140
};\addlegendentry{$\lambda_\text{max}=\infty$, $b_i$ via \eqref{eq:bias_term}}

\end{semilogyaxis}

\end{tikzpicture}
    \vspace{-2mm}
    \caption{ANV.}
\end{subfigure}
\caption{FER/\ac{ANV} vs. $E_b/N_0$ over the \ac{biAWGN} channel for \edd{the} $\left(128,64\right)$ \ac{PAC} code under \ac{SCOS} decoding with various bias terms and maximum complexity constraints \eddd{such that $\eta = \lambda_{\text{max}}$}.}
\label{fig:scos_dfpolar_robust}
\end{figure*}

\subsection{\ed{SC Ordered Search Decoding with Maximum Path Metric}}
\label{sec:SCOS_max_PM}
Monte Carlo simulation \edd{under} genie-aided \ac{SC} decoding~\cite{arikan2009channel} can be used to approximate the \ac{PDF} of the \ac{PM} for the transmitted message at a given \ac{SNR}. \ed{Note that $M\left(u^N\right)$ is a \ac{RV} where the source of randomness is the channel output. Since we consider symmetric \acp{B-DMC} and a linear code with uniform distribution, the \ac{PDF} of $M\left(u^N\right)$ could be computed with an all-zero codeword assumption.} \ed{For instance,} Figure~\ref{fig:dfpolar_PM} provides the \ac{PDF} for the $(128,64)$ \ac{PAC} code at $E_b/N_0 = 3.5$~dB. Observe that $\text{Pr}\left(M\left(u^N\right) > 50 \right) \approx 0$, i.e., if the decoder discards the paths having \acp{PM} larger than $50$\edd{,} then the performance degradation is negligible while reducing computational complexity. Such a modification is particularly relevant when a maximum complexity constraint is imposed on \ac{SCOS} decoding. In this case, unnecessary node-visits drain the computation budget and increase the number of suboptimal decisions. Moreover, the threshold test lets the decoder reject unreliable decisions and reduces the number of undetected errors if the threshold is carefully optimized, see\cite{forney1968exponential}.
\begin{figure}[t]
	\centering
	\footnotesize
	\begin{tikzpicture}[scale=1]
\begin{axis}[
ymin=0,
ymax=0.09,
width=3.5in,
height=1.5in,
xtick={0,5,...,50},
grid=both,
xmin = 0,
xmax = 50,
xlabel = $M\left(u^N\right)$,
ylabel = Empirical PDF
]

\addplot[red,thick,each nth point=1]
table[x=snr,y=fer]{snr fer
	4.5 3.937e-06
	5.008 1.7717e-05
	5.516 3.7402e-05
	6.024 5.315e-05
	6.532 0.0001752
	7.04 0.00037795
	7.548 0.0006752
	8.056 0.0012126
	8.564 0.0019173
	9.072 0.0029154
	9.58 0.0043917
	10.088 0.0061732
	10.596 0.0083346
	11.104 0.011596
	11.612 0.014433
	12.12 0.01811
	12.628 0.022766
	13.136 0.027742
	13.644 0.031711
	14.152 0.037126
	14.66 0.042079
	15.168 0.047419
	15.676 0.052772
	16.184 0.056998
	16.692 0.060994
	17.2 0.06436
	17.708 0.068354
	18.216 0.071431
	18.724 0.072797
	19.232 0.073665
	19.74 0.07428
	20.248 0.073201
	20.756 0.072776
	21.264 0.071417
	21.772 0.070028
	22.28 0.06737
	22.788 0.064431
	23.296 0.062002
	23.804 0.058547
	24.312 0.055217
	24.82 0.051171
	25.328 0.047766
	25.836 0.044037
	26.344 0.039906
	26.852 0.036752
	27.36 0.033661
	27.868 0.030234
	28.376 0.027197
	28.884 0.023583
	29.392 0.021573
	29.9 0.019181
	30.408 0.016947
	30.916 0.015
	31.424 0.012915
	31.932 0.011522
	32.44 0.0096673
	32.948 0.0084154
	33.456 0.0073701
	33.964 0.0064587
	34.472 0.0055394
	34.98 0.0046713
	35.488 0.004
	35.996 0.0034213
	36.504 0.0029744
	37.012 0.0022579
	37.52 0.0018799
	38.028 0.0016457
	38.536 0.001435
	39.044 0.0011732
	39.552 0.00096457
	40.06 0.00081693
	40.568 0.00062992
	41.076 0.00052559
	41.584 0.0004311
	42.092 0.00035236
	42.6 0.00032874
	43.108 0.00025787
	43.616 0.00019291
	44.124 0.0001811
	44.632 8.2677e-05
	45.14 8.4646e-05
	45.648 6.1024e-05
	46.156 7.874e-05
	46.664 3.937e-05
	47.172 5.1181e-05
	47.68 2.5591e-05
	48.188 2.9528e-05
	48.696 1.7717e-05
	49.204 1.1811e-05
	49.712 2.3622e-05
	50.22 1.5748e-05
	50.728 1.378e-05
	51.236 5.9055e-06
	51.744 9.8425e-06
	52.252 1.9685e-06
	52.76 1.9685e-06
	53.268 1.9685e-06
	53.776 0
	54.284 1.9685e-06
	54.792 1.9685e-06
	55.3 0
	
};\addlegendentry{no approx.}

\addplot[blue,thick,each nth point=2]
table[x=snr,y=fer]{snr fer
	0 0.001643
	0.437 0.002778
	0.874 0.0046659
	1.311 0.0071899
	1.748 0.010387
	2.185 0.014114
	2.622 0.018254
	3.059 0.023183
	3.496 0.028588
	3.933 0.0337
	4.37 0.039341
	4.807 0.04586
	5.244 0.050435
	5.681 0.056446
	6.118 0.061201
	6.555 0.066089
	6.992 0.069796
	7.429 0.073314
	7.866 0.076053
	8.303 0.078675
	8.74 0.079719
	9.177 0.08122
	9.614 0.080799
	10.051 0.081057
	10.488 0.080153
	10.925 0.077897
	11.362 0.076657
	11.799 0.072856
	12.236 0.070668
	12.673 0.068167
	13.11 0.064941
	13.547 0.061497
	13.984 0.05719
	14.421 0.054027
	14.858 0.050549
	15.295 0.047121
	15.732 0.043009
	16.169 0.039728
	16.606 0.036281
	17.043 0.033558
	17.48 0.03065
	17.917 0.02759
	18.354 0.024522
	18.791 0.022229
	19.228 0.020018
	19.665 0.017927
	20.102 0.015867
	20.539 0.014066
	20.976 0.01246
	21.413 0.011275
	21.85 0.0097963
	22.287 0.0084073
	22.724 0.0075629
	23.161 0.0065904
	23.598 0.0058124
	24.035 0.005
	24.472 0.0043753
	24.909 0.0037071
	25.346 0.0033593
	25.783 0.0027689
	26.22 0.0023822
	26.657 0.0020618
	27.094 0.0017551
	27.531 0.0015057
	27.968 0.0012952
	28.405 0.0010481
	28.842 0.00091533
	29.279 0.00073913
	29.716 0.00063844
	30.153 0.00051945
	30.59 0.00043707
	31.027 0.00035011
	31.464 0.00027918
	31.901 0.00024027
	32.338 0.00025858
	32.775 0.00016705
	33.212 0.00018993
	33.649 0.00012815
	34.086 0.00012357
	34.523 0.00010526
	34.96 8.9245e-05
	35.397 4.5767e-05
	35.834 5.2632e-05
	36.271 3.8902e-05
	36.708 2.5172e-05
	37.145 2.9748e-05
	37.582 1.373e-05
	38.019 1.373e-05
	38.456 2.746e-05
	38.893 9.1533e-06
	39.33 1.6018e-05
	39.767 1.373e-05
	40.204 2.2883e-06
	40.641 2.2883e-06
	41.078 1.1442e-05
	41.515 0
	41.952 2.2883e-06
	42.389 2.2883e-06
	42.826 2.2883e-06
	43.263 0
	43.7 0
	
};\addlegendentry{min-sum approx.}

\end{axis}
\end{tikzpicture}
	\caption{Empirical \ac{PDF} (via $10^7$~samples) of the $M\left(u^N\right)$ over the \ac{biAWGN} channel at $3.5$~dB \ac{SNR} for \edd{the} $\left(128,64\right)$ \ac{PAC} code. The \edd{both curves are} obtained via genie-aided \ac{SC} decoding~\cite{arikan2009channel}\edd{, where the blue line uses the} min-sum approximation \edd{and} the red line without any approximation.}
	\label{fig:dfpolar_PM}
\end{figure}
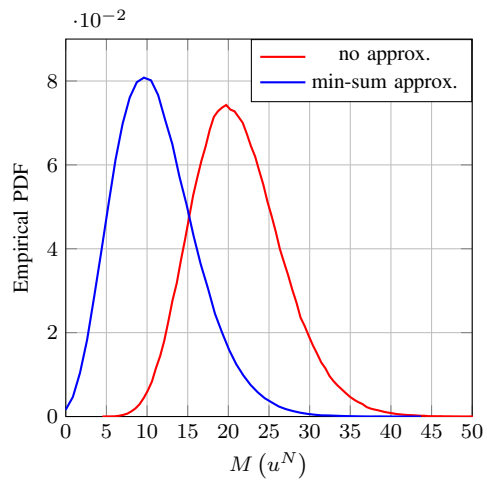
In the following, we modify \ac{SCOS} decoding by setting a maximum \ac{PM} $M_\text{max}$ as shown in Algorithm~\ref{alg:SCOSwPMmax}. 
\begin{algorithm}[t]
        \footnotesize
	\SetNoFillComment
	\DontPrintSemicolon
	\SetKwInOut{Input}{Input}\SetKwInOut{Output}{Output}
	\Input{input \ac{LLR}s $\ell^N$, \textcolor{red}{$M_\text{max}$}}
	\Output{output vector $\hat{\texttt{u}}$, \textcolor{red}{decoding state $\omega$}}
	\BlankLine
	$\mathcal{L}=\varnothing, \mathcal{E}_\texttt{p}=\varnothing, \texttt{M}_\texttt{cml}=\textcolor{red}{M_\text{max}}$, \textcolor{red}{$\omega=0$}\\
	\For{$i=1,2,\dots,N$}{
		$\texttt{L}\left[1,i\right]=\ell_i$
	}
	$\textcolor{red}{i_\text{end} =~} \text{SCDec}\left(1,\varnothing\right)$\\
	\textcolor{red}{
		\lIf{$i_\text{end} = N$}{$\omega=1$}
	}
	\For{$i=1,2,\dots,N$}{
		\If{$i\in\mathcal{A}$ and $\overline{\texttt{M}}\left[i\right]<\texttt{M}_\texttt{cml}$}{
			$\text{Insert\edddd{Heap}}\left(\left\langle\{i\},\overline{\texttt{M}}\left[i\right],\overline{\texttt{S}}\left[i\right]\right\rangle\right)$\\
		}
	}
	\While{$\mathcal{L}\neq\varnothing$}{
		$\left\langle \mathcal{E},\overline{\texttt{M}}_\mathcal{E},\overline{\texttt{S}}_\mathcal{E}\right\rangle= \text{pop\edddd{Min}}\left(\mathcal{L}\right)$\\
		\If{$\overline{\texttt{M}}_\mathcal{E} < \texttt{M}_\texttt{cml}$}{
			$i_\text{start}=\text{FindStartIndex}\left(\mathcal{E},\mathcal{E}_\texttt{p}\right)$\\
			$i_\text{end}=\text{SCDec}\left(i_\text{start},\mathcal{E}\right)$\\
			\textcolor{red}{
				\lIf{$i_\text{end} = N$}{$\omega=1$}
			}
			\For{$i=\text{maximum}\left(\mathcal{E}\right)+1,\dots,i_\text{end}$}{
				\If{$i\in\mathcal{A}$ and $\overline{\texttt{M}}\left[i\right]<\texttt{M}_\texttt{cml}$}{
					$\text{Insert\edddd{Heap}}\left(\left\langle \mathcal{E}\cup\{i\},\overline{\texttt{M}}\left[i\right],\overline{\texttt{S}}\left[i\right]\right\rangle\right)$\\
				}
			}
			$\mathcal{E}_\texttt{p} = \mathcal{E}$
		}
	}
	\Return{$\hat{\texttt{u}}$, \textcolor{red}{$\omega$}}
	\caption{$\text{SCOS \textcolor{red}{with maximum PM}}\left(\ell^N, \textcolor{red}{M_\text{max}}\right)$}
	\label{alg:SCOSwPMmax}
\end{algorithm}

The differences in \ac{SCOS} decoding with a maximum \ac{PM} are highlighted in red. The algorithm requires the input $M_\text{max}$ to discard candidates with \acp{PM} larger than $M_\text{max}$, where the \ac{PM} of the current most-likely path is initialized to this threshold (\texttt{line 1}). In addition, the algorithm has a binary output $\omega$, which is initialized to $0$ (\texttt{line 1}), is set to $1$ if any estimate with \ac{PM} less than $M_\text{max}$ is found (\texttt{lines 5 and 14}) and stays as $0$ otherwise. Note that the initial \ac{SC} decoding does not reach a leaf node in case SC decoding path is exceeding $M_\text{max}$ (\texttt{line 4}).

\edd{Figure~\ref{fig:scos_dfpolar_PMmax} compares the performance of \ed{SCOS decoding with maximum \ac{PM} constraint} (Algorithm \ref{alg:SCOSwPMmax}) to that of original \ac{SCOS} decoding. The former gains $\approx0.2$~dB with the same maximum complexity constraint \eddd{$\lambda_\text{max}=5$} if $M_\text{max}=35$. Note that the average complexity is similar}.
\begin{figure}
\centering
\begin{subfigure}{0.49\textwidth}
\centering
    \begin{tikzpicture}[scale=1]
\footnotesize
\begin{semilogyaxis}[
legend pos=south west,
legend columns=1,
legend cell align=left,
ymin=0.0001,
ymax=0.4,
width=3.25in,
height=2in,
grid=both,
xmin = 1.0,
xmax = 4.0,
xlabel = $E_b/N_0$ in dB,
ylabel = FER,
xtick={1,1.5,2,2.5,3,3.5,4,4.5,5,5.5}
]

\addplot[red, mark = o]
table[x=snr,y=fer]{snr fer
1 0.412
1.25 0.3305
1.5 0.229
1.75 0.144
2 0.088
2.25 0.05425
2.5 0.035
2.75 0.02
3 0.0094091
3.25 0.0047778
3.5 0.0019057
3.75 0.00073529
4 0.00027701
};\addlegendentry{$M_\text{max}=\infty$}


\addplot[blue, mark = square]
table[x=snr,y=fer]{snr fer
1 0.3355
1.25 0.266
1.5 0.19
1.75 0.128
2 0.08225
2.25 0.044833
2.5 0.026
2.75 0.014
3 0.0066875
3.25 0.0028571
3.5 0.00096154
3.75 0.00036679
4 0.0001073
};\addlegendentry{$M_\text{max}=35$}

\addplot[brown, mark = x]
table[x=snr,y=fer]{snr fer
1 0.3425
1.25 0.2565
1.5 0.196
1.75 0.142
2 0.08535
2.25 0.053167
2.5 0.035667
2.75 0.018417
3 0.0085833
3.25 0.0046818
3.5 0.001627
3.75 0.00055866
4 0.00017857
};\addlegendentry{$M_\text{max}=50$}

\addplot[black, mark = +]
table[x=snr,y=fer]{snr fer
1 0.319
1.25 0.252
1.5 0.176
1.75 0.1225
2 0.08675
2.25 0.05775
2.5 0.040333
2.75 0.026875
3 0.016
3.25 0.0092273
3.5 0.0068667
3.75 0.0035893
4 0.0018518
};\addlegendentry{$M_\text{max}=25$}

\end{semilogyaxis}

\end{tikzpicture}
    \vspace{-2mm}
    \caption{Performance.}
\end{subfigure}
\begin{subfigure}{0.49\textwidth}
\centering
\hspace{6pt}
    \begin{tikzpicture}[scale=1]
\footnotesize
\begin{axis}[
width=3.25in,
height=1.5in,
ymin=1,
ymax=5,
grid=both,
xmin = 1,
xmax = 4,
xlabel = $E_\text{b}/N_0\ \text{in dB}$,
ylabel = {$\mathbb{E}\left[\Lambda\right]$},
xtick={1,1.5,2,2.5,3,3.5,4}
]
\addplot[black, mark = +, dashed, mark options = solid]
table[x=snr,y=fer]{snr fer
1 4.4516
1.25 4.1554
1.5 3.7927
1.75 3.346
2 2.8746
2.25 2.4682
2.5 2.1414
2.75 1.7861
3 1.5044
3.25 1.2962
3.5 1.1744
3.75 1.0963
4 1.0506
};\addlegendentry{$M_\text{max}=25$}

\addplot[red, mark = o, dashed, mark options = solid]
table[x=snr,y=fer]{snr fer
1 4.5008
1.25 4.1815
1.5 3.8308
1.75 3.4375
2 3.0013
2.25 2.4949
2.5 2.161
2.75 1.7834
3 1.5367
3.25 1.3314
3.5 1.1937
3.75 1.1094
4 1.0561
};\addlegendentry{$M_\text{max}=\infty$}

\addplot[brown, mark = x, dashed, mark options = solid]
table[x=snr,y=fer]{snr fer
1 4.4463
1.25 4.1388
1.5 3.7876
1.75 3.4523
2 2.9205
2.25 2.513
2.5 2.1693
2.75 1.793
3 1.5231
3.25 1.3319
3.5 1.1898
3.75 1.1058
4 1.0552
};\addlegendentry{$M_\text{max}=50$}

\addplot[blue, mark = square, dashed, mark options = solid]
table[x=snr,y=fer]{snr fer
1 4.4742
1.25 4.1769
1.5 3.8038
1.75 3.3113
2 2.931
2.25 2.5124
2.5 2.1198
2.75 1.7781
3 1.5189
3.25 1.3131
3.5 1.1875
3.75 1.1037
4 1.0542
};\addlegendentry{$M_\text{max}=35$}

\end{axis}

\end{tikzpicture}
    \vspace{-2mm}
    \caption{ANV.}
\end{subfigure}
\caption{FER/\ac{ANV} vs. $E_b/N_0$ over the \ac{biAWGN} channel for \edd{the} $\left(128,64\right)$ \ac{PAC} code under \ac{SCOS} decoding with various maximum \acp{PM} \eddd{and fixed maximum complexity constraints $\lambda_\text{max}=\eta=5$}.}
\label{fig:scos_dfpolar_PMmax}
\end{figure}
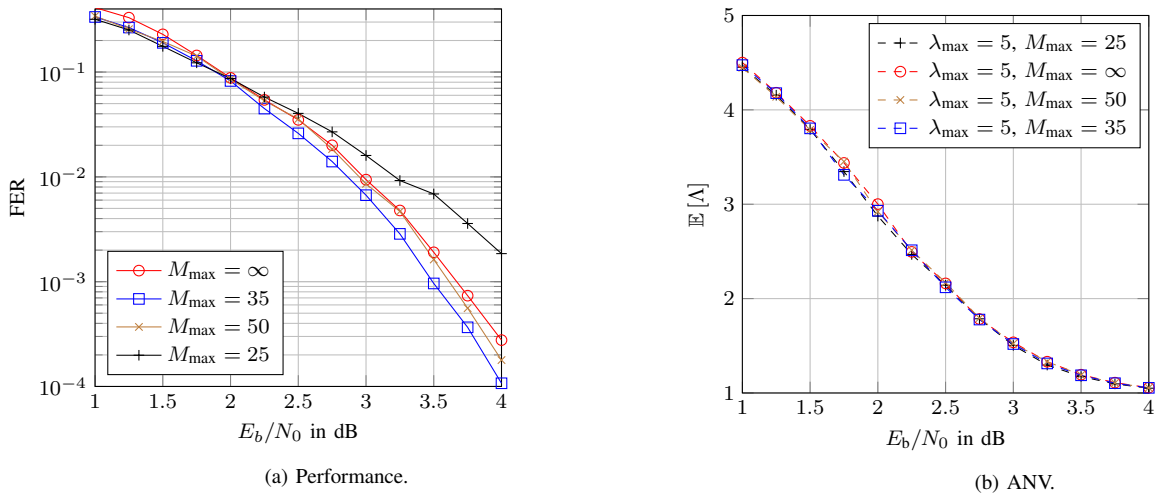

\edd{Observe now that the proposed modification enables \ac{SCOS} decoding to reject an \emph{unreliable} estimate inherently.} For a given threshold $M_\text{max}$, define the binary \ac{RV}
\begin{equation}\label{eq:rv_undetected}
\Omega = \mathbbm{1}\left\{M(\hat{u}^N)\leq M_\text{max}\right\}
\end{equation}
where the indicator function $\mathbbm{1}\{\mathrm{P}\}$ takes on the value $1$ if the proposition $\mathrm{P}$ is true and $0$ otherwise. The proposition of the indicator function \eqref{eq:rv_undetected} reads as ``the modified \ac{SCOS} decoding finds an estimate $\hat{u}^N$ with a \ac{PM} smaller than $M_\text{max}$''. The undetected error probability of the algorithm is
\begin{align}\label{eq:undetected_error_prob}
\text{Pr}\left(\hat{U}^N\neq U^N,\Omega = 1\right).
\end{align}
The overall error probability is the sum of the detected and undetected error probabilities, i.e., we have
\begin{align}\label{eq:overall_error_prob}
\text{Pr}\left(\hat{U}^N\neq U^N\right) = \sum_{\omega\in\{0,1\}}\text{Pr}\left(\hat{U}^N \neq U^N, \Omega = \omega\right)
\end{align}
which follows from the law of total probability. The parameter $M_\text{max}$ controls the FER and \ac{uFER} tradeoff~\cite{forney1968exponential,hof2010performance}. In particular, \eqref{eq:undetected_error_prob} is the \ac{LHS} of \eqref{eq:overall_error_prob} if $M_\text{max} = \infty$. \edd{Numerical results illustrating benefits in \ac{uFER} is provided towards the end of the next section.}

\section{Comparison to Existing Polar Decoders}
\label{sec:numerical}
This section compares the proposed \ac{SCOS} decoding to \ac{SC-Fano}, \ac{SCL}\edddd{, \ac{SCS} as well as \ac{ORBGRAND}\cite{DAM22:ORBGRAND}} algorithms in the short block length regime, e.g., for $N=128$, for two different code dimensions $K=29$ and $K=99$, over the \ac{biAWGN} channels. \eddd{Note that the score is given by \eqref{eq:Score2} for \ac{SCOS}, \ac{SC-Fano}, and \ac{SCS} decoding algorithms, where we use min-sum approximation for the first term and quantized density evolution for the bias as explained in Section \ref{sec:bias}.} We provide \acp{FER} and average complexity scores computed via \eqref{eq:complexityscore} for these codes.

\edd{Figures \ref{fig:compscos_dfRMcodes}(a) and \ref{fig:compscos_dfRMcodes}(b) provides the results for the $(128,29)$ code, where \ac{SCOS} with maximum \eddd{complexity constraints} $\lambda_{\mathrm{max}}\eddd{\,=\eta}=64$ approaches to its \ac{ML} performance outperforming other decoding algorithms.
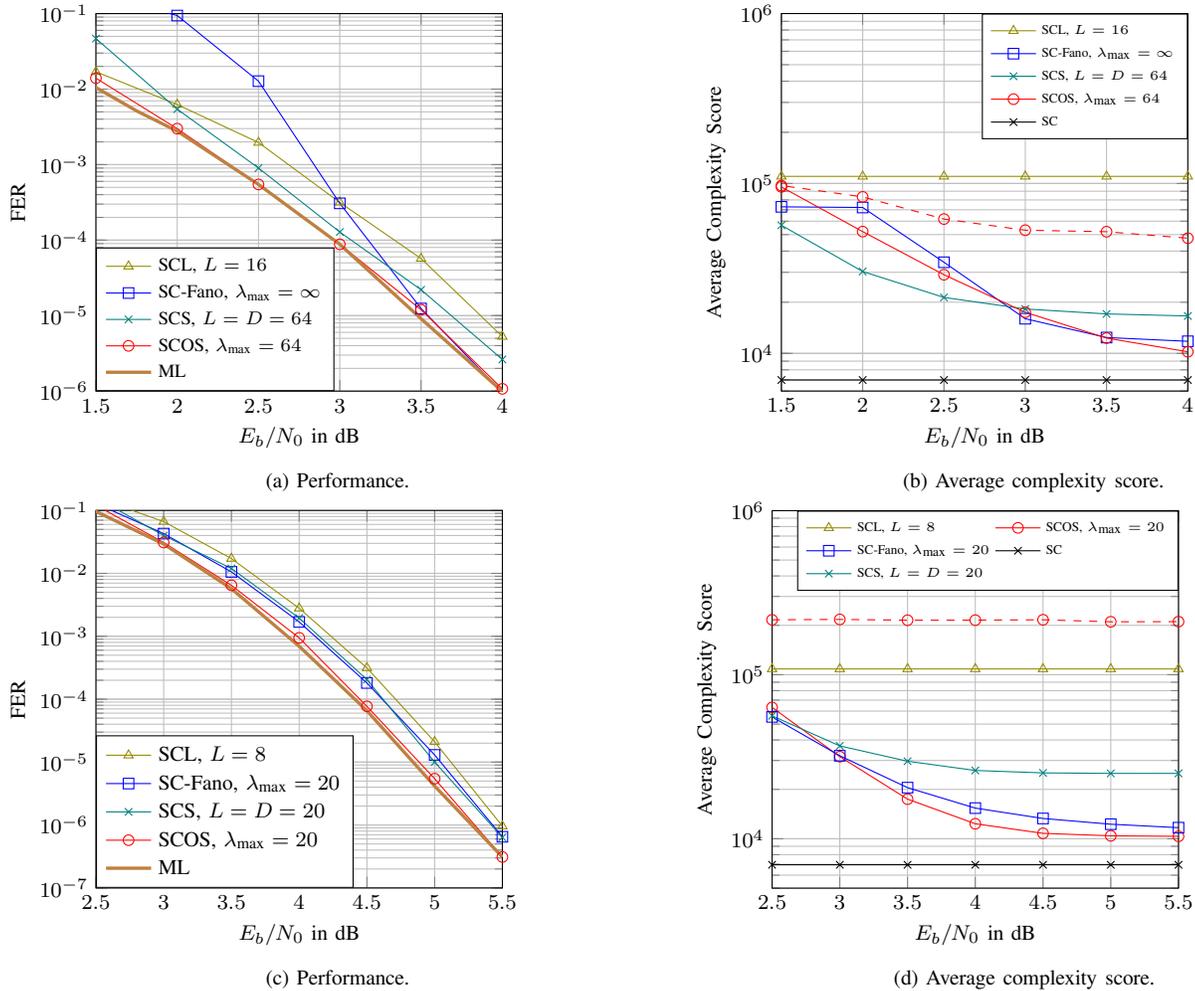
\begin{figure*}
\centering
\begin{subfigure}{0.49\textwidth}
    \begin{tikzpicture}[scale=1]
\footnotesize
\begin{semilogyaxis}[
legend cell align=left,
legend pos=south west,
legend cell align=left,
ymin=0.000001,
ymax=0.1,
width=3.5in,
height=2.1in,
grid=both,
xmin = 1.5,
xmax = 4,
xlabel = $E_b/N_0$ in dB,
ylabel = FER,
xtick={1,1.5,2,2.5,3,3.5,4,4.5,5,5.5}
]

\addplot[olive, mark = triangle]
table[x=snr,y=fer]{snr fer
1.00      0.050000  
1.50      0.017000  
2.00      0.006250   
2.50      0.001962   
3.00      0.000316   
3.50      0.000057  
4.00      0.00000528
};\addlegendentry{\tiny{\ac{SCL}, $L=16$}}

\addplot[blue, mark = square]
table[x=snr,y=fer]{snr fer
1.00      0.292000  
1.50      0.184000  
2.00      0.094250   
2.50      0.012750  
3.00      0.000307   
3.50      0.00001242
4.00      9.503444998812069e-7
};\addlegendentry{\tiny{SC-Fano, $\lambda_\text{max}=\infty$}}

\addplot[teal, mark = x]
table[x=snr,y=fer]{snr fer
1.00      0.195695    
1.50      0.046664     
2.00      0.005432     
2.50      0.000901     
3.00      0.000128
3.50      0.000022
4.00      2.62896859E-6
};\addlegendentry{\tiny{\ac{SCS}, $L=D=64$}}

\addplot[red, mark = o]
table[x=snr,y=fer]{snr fer
1.00      0.041425    
1.50      0.013959     
2.00      0.002993     
2.50      0.000546     
3.00      0.000088    
3.50      1.19714586e-5
4.00      1.07266859e-6
};\addlegendentry{\tiny{\ac{SCOS}, $\lambda_\text{max}\eddd{=\eta}=64$}}

\addplot[brown, line width = 1.25pt]
table[x=snr,y=fer]{snr fer
1.00      0.027750    
1.25      0.016375    
1.50      0.010500    
1.75      0.005150    
2.00      0.002775    
2.50      0.000546     
3.00      0.000088  
3.50      9.10415895e-6
3.75      3.006253e-06
4.00      9.868551e-07
};\addlegendentry{\tiny{\ac{ML}}} 

\end{semilogyaxis}
\end{tikzpicture}
    \vspace{-2mm}
    \caption{Performance.}
\end{subfigure}
\hfill
\begin{subfigure}{0.49\textwidth}
    \begin{tikzpicture}[scale=1]
\footnotesize
\begin{semilogyaxis}[
legend style={at={(0.98,0.875)},anchor=north east},
legend columns = 2,
legend cell align=left,
ymin=6e3, ymax=1e6,
width=3.5in,
height=2.1in,
grid=both,
xmin = 1.5,
xmax = 4,
xlabel = $E_b/N_0$ in dB,
ylabel = Average Complexity Score,
xtick={1,1.5,2,2.5,3,3.5,4,4.5,5,5.5}
]

\addplot[olive, mark = triangle]
table[x=snr,y=fer]{snr fer
1.00      110091.0
1.50      110091.0
2.00      110091.0
2.50      110091.0
3.00      110091.0
3.50      110091.0
4.00      110091.0
4.50      110091.0
};\addlegendentry{\tiny{SCL}, $L=16$}

\addplot[blue, mark = square]
table[x=snr,y=fer]{snr fer
1.00     56947.825500
1.50     72854.515000
2.00     72210.971750
2.50     34366.900375
3.00     15992.678232
3.50     12394.543332
4.00     11774.796544
};\addlegendentry{\tiny{\ac{SC-Fano}, $\lambda_\text{max}=\infty$}}

\addplot[red, mark = o]
table[x=snr,y=fer]{snr fer
1.50     84594.907270 
2.00     47706.454102
2.50     27240.606117
3.00     16289.673653
3.50     11524.768718 
4.00     9598.631067
};\addlegendentry{\tiny{\ac{SCOS}, $\lambda_\text{max}\eddd{=\eta}=64$}} 

\addplot[black, mark = x]
table[x=snr,y=fer]{snr fer
	1 6969
	1.5 6969
	2 6969
	2.5 6969
	3 6969
	3.5 6969
	4.0 6969
	4.5 6969
};\addlegendentry{\tiny{\ac{SC}}}

\addplot[teal, mark = x]
table[x=snr,y=fer]{snr fer
1.50     23824.730745
2.00     15645.124738
2.50     12778.733160
3.00     11641.315235
3.50     11136.498623
4.00     10889.186038
};\addlegendentry{\tiny{\ac{SCS}, $L=D=64$}}

\addplot[red, mark = o, dashed, mark options = solid]
table[x=snr,y=fer]{snr fer
	1.5 589151.686387
	2 582978.900215
	2.5 582558.197977
	3 582926.133546
    3.5 582280.613331
    4.0 577323.172826
};

\end{semilogyaxis}
\end{tikzpicture}
    \vspace{-2mm}
    \caption{Average complexity score.}
\end{subfigure}
\begin{subfigure}{0.49\textwidth}
    \begin{tikzpicture}[scale=1]
\footnotesize
\begin{semilogyaxis}[
legend pos = south west,
legend cell align=left,
ymin=0.0000001,
ymax=0.1,
width=3.5in,
height=2.5in,
grid=both,
xmin = 2.5,
xmax = 5.5,
xlabel = $E_b/N_0$ in dB,
ylabel = FER,
xtick={1,1.5,2,2.5,3,3.5,4,4.5,5,5.5}
]

\addplot[cyan, mark = *]
table[x=snr,y=fer]{snr fer
2.50      0.322581
3.00      0.131406     
3.50      0.047059    
4.00      0.009549    
4.50  	  0.001632
5.00  	  0.000150
5.50  	  0.000016
};\addlegendentry{\tiny{ORBGRAND}, TEP~$\leq5\cdot 10^6$}

\addplot[olive, mark = triangle]
table[x=snr,y=fer]{snr fer
2.50      0.160000
3.00      0.066500     
3.50      0.017167 
4.00      0.002789    
4.50  	  0.000314
5.00      0.000021
5.50      0.00000096
};\addlegendentry{\tiny{\ac{SCL}, $L=8$}}

\addplot[blue, mark = square]
table[x=snr,y=fer]{snr fer
2.50    0.129000
3.00 	0.042833  
3.50 	0.010550  
4.00 	0.001692  
4.50 	0.000181
5.00    0.000013
5.50    0.00000065
};\addlegendentry{\tiny{\ac{SC-Fano}, $\lambda_\text{max}=20$}}

\addplot[teal, mark = x]
table[x=snr,y=fer]{snr fer
2.50      0.159236
3.00      0.039920     
3.50      0.011973 
4.00      0.001950    
4.50  	  0.000205
5.00      0.000010
5.50      6.35678313E-7
};\addlegendentry{\tiny{\ac{SCS}, $L=D=20$}}

\addplot[red, mark = o]
table[x=snr,y=fer]{snr fer
2.50      0.120773
3.00      0.030950     
3.50      0.006448    
4.00      0.000944    
4.50  	  0.000077
5.00  	  5.45e-06
5.50  	  3.104857e-07
};\addlegendentry{\tiny{\ac{SCOS}, $\lambda_\text{max}\eddd{=\eta}=20$}}

\addplot[brown, line width = 1.25pt]
table[x=snr,y=fer]{snr fer
2.25      0.150250     
2.50      0.096250     
3.00      0.029250     
3.50      0.005700    
4.00      0.000694    
4.50  	  6.561680e-05
5.00  	  4.162504e-06
5.50  	  3.104857e-07
};\addlegendentry{\tiny{\ac{ML}}}

\end{semilogyaxis}
\end{tikzpicture}
    \vspace{-6mm}
    \caption{Performance.}
\end{subfigure}
\hfill
\begin{subfigure}{0.49\textwidth}
    \begin{tikzpicture}[scale=1]
\footnotesize
\begin{semilogyaxis}[
legend cell align=left,
legend columns = 2,
ymin=5e3, ymax=1e6,
width=3.5in,
height=2.5in,
grid=both,
xmin = 2.5,
xmax = 5.5,
xlabel = $E_b/N_0$ in dB,
ylabel = Average Complexity Score,
xtick={1,1.5,2,2.5,3,3.5,4,4.5,5,5.5}
]

\addplot[olive, mark = triangle]
table[x=snr,y=fer]{snr fer
2.50     108435
3.00     108435
3.50     108435
4.00     108435
4.50     108435
5.00     108435
5.50     108435
};\addlegendentry{\tiny{SCL, $L=8$}}

\addplot[blue, mark = square]
table[x=snr,y=fer]{snr fer
2.50     55334.778500
3.00     32064.443000
3.50     20496.849500
4.00     15379.700650
4.50     13296.282026
5.00     12266.276545
5.50     11684.499900
};\addlegendentry{\tiny{\ac{SC-Fano}, $\lambda_\text{max}=20$}}

\addplot[red, mark = o]
table[x=snr,y=fer]{snr fer
2.50     45975.294749
3.00     25731.978537
3.50     15262.849755
4.00     11187.683122
4.50     9950.460337
5.00     9642.990532
5.50     9576.958774
};\addlegendentry{\tiny{\ac{SCOS}, $\lambda_\text{max}\eddd{=\eta}=20$}}

\addplot[black, mark = x]
table[x=snr,y=fer]{snr fer
2.50     6942
3.00     6942
3.50     6942
4.00     6942
4.50     6942
5.00     6942
5.50     6942
};\addlegendentry{\tiny{\ac{SC}}}

\addplot[teal, mark = x]
table[x=snr,y=fer]{snr fer
2.50     28500.086538
3.00     19970.064587
3.50     15955.121501
4.00     14132.948000
4.50     13612.019957
5.00     13504.327791
5.50     13480.195601
};\addlegendentry{\tiny{\ac{SCS}, $L=D=20$}}

\addplot[red, mark = o, dashed, mark options = solid]
table[x=snr,y=fer]{snr fer
2.50     207226.111259
3.00     202368.539774
3.50     205677.160163
4.00     203259.981340
4.50     201611.807720
5.00     204178.787029
5.50     204178.787029
};

\end{semilogyaxis}
\end{tikzpicture}
    \vspace{-2mm}
    \caption{Average complexity score.}
\end{subfigure}
\caption{\edd{FER and average complexity score vs. $E_b/N_0$ over the \ac{biAWGN} channel for $(128,29)$ (a-b) and $(128,99)$ (c-d) \ac{dRM} codes under \ac{SCOS} compared to \ac{SC}-Fano, \ac{SCL} and \ac{SCS} decoding algorithms under various complexity constraints. As reference, the worst-case complexity scores obtained via simulations for each SNR value are also provided for SCOS as dashed curves with solid marks.}}
\label{fig:compscos_dfRMcodes}
\end{figure*}
This performance is achieved with the lowest complexity at high SNR. Its performance is followed by \ac{SCS} decoding with stack size $D=64$ and maximum node-visits $L=64$ in a wide SNR regime robustly with $0.25$ dB difference. It provides lower complexity compared to \ac{SCOS} and \ac{SC-Fano} decoding algorithms at SNR values smaller than \edddd{$3.5$} dB; however, it requires more complexity at high\edddd{er} SNR values with worse performance. \edddd{Although} its space complexity is \edddd{already} larger than \eddd{other competitors for the chosen parameters}\edddd{, the performance of \ac{SCS} decoding can be further improved if the stack size is chosen to be at least a few times of $L$, which is, e.g., chosen to be $D=LN$ in \cite[Sec. IV]{trifonov2018score}. However, this comes at the expense of also potentially higher computational especially at low and mid \ac{SNR} ranges.} \ac{SC-Fano} decoding (even with unbounded complexity) does not perform well in particular at relatively low and medium \acp{FER}, which might require a careful optimization of the search parameter $\Delta$. Nevertheless, it is very competitive to \ac{SCOS} decoding at low \acp{FER} in performance and average complexity score. In addition, it does not require large space complexity as it is the case for \ac{SCS} decoding. Also for the $(128,99)$ code, \ac{SCOS} decoding with $\lambda_{\mathrm{max}}\eddd{\,=\eta}=20$ approaches \ac{ML} decoding tightly (see Figure \ref{fig:compscos_dfRMcodes}(c)). Figure \ref{fig:compscos_dfRMcodes}(d) shows that its average complexity score is lower than all the competitors for \ac{SNR} values larger than \edddd{$3.5$} dB, i.e., except for high \acp{FER}. Its complexity score is within $\edddd{1.4}$ times of SC decoding at high SNR values in both cases. \edddd{Although the complexity of SCOS decoding, at least in the proxy of \ac{ANV}, is reaching to that of SC decoding at high SNR values,it requires extra operations similar to other improved decoding algorithms, e.g., in \texttt{lines 12-14} of Algorithm \ref{alg:SCDec} for the computation of scores and \acp{PM} independent of operating \ac{SNR}, which hinders its complexity score to hit that of SC decoding at high \ac{SNR} values. In the case of SCS decoding, this is partially due to the push and pop operations as well as for the calculation of the scores for the pushed paths. To complement the numerical results of Figure \ref{fig:compscos_dfRMcodes}, we provide Table \ref{tab:dRM_scs} for SCS decoding, where the parameters of the decoder is kept the same as in the figure.} As expected, \ac{SCL} decoding requires much larger complexity compared to all \edddd{SC-based} complexity-adaptive decoding algorithms. In addition, we also provided the worst-case complexity scores obtained during simulations at each \ac{SNR} value for the proposed \ac{SCOS} algorithm, \edddd{which are a few times more than that of SCL decoding due to the higher $\lambda_{\text{max}}$ of the latter}. \edddd{Note that the performance of the $(128, 99)$ code is also provided under \ac{ORBGRAND}, which is an efficient approximation of soft GRAND\cite{SDM20}, where the maximum number of \acp{TEP} is set to $5\times10^6$\eddddd{. Observe that it} performs within $0.7$ dB from the \ac{ML} decoding performance \eddddd{for the considered case}. The average complexity score\eddddd{\footnote{\eddddd{The codeword membership test costs a binary vector-matrix multiplication using the parity-check matrix of the underlying dRM code similar to \cite{YDG+22}. This way of membership test is not mandatory for GRAND-based decoders and more efficient methods might be used depending on the underlying code, e.g., applying a polar transform is sufficient for polar and RM codes.}}} reaches roughly to $5\times10^6$ at an SNR of $E_b/N_0=5.5$ dB, which we skip in Figure \ref{fig:compscos_dfRMcodes}(d) as it is an order of magnitude larger than that of SCL decoding with $L=20$.}} \eddddd{As Remark 2 hints, SCOS decoding provides advantage over universal decoders like ORBGRAND if the underlying code is somewhat suited for SC-based decoders.}
\begin{table}
	\caption{\edddd{dRM codes under SCS decoding}}
\vspace{-3mm}	
 \begin{center}
	\begin{tabular}{ccccccc}
		\hline\hline
		$(N,K)$ & $E_b/N_0$ & Addition & Compare & XOR & \eqref{eq:complexityscore} & SC \eqref{eq:complexityscore} \\
		\hline
		$(128,29)$ & $4.00$ & $771$ & $613$ & $1040$ & $10889$ & $6956$ \\
		$(128,99)$ & $5.50$ & $902$ & $852$ & $1149$ & $13480$ & $6945$ \\
		\hline\hline
	\end{tabular}
\end{center}
	\label{tab:dRM_scs}
\end{table}
\renewcommand{\arraystretch}{1}

\edd{Finally,} Figure~\ref{fig:uFER} illustrates that the $(128,64)$ \ac{PAC} code under modified \ac{SCOS} decoding gains in overall \ac{FER} and in \ac{uFER} as compared to a $(128,71)$ polar code concatenated with a \ac{CRC}-$7$ (resulting in a $(128,64)$ overall code) under \ac{SCL} decoding with $L=16$ at high \ac{SNR}. Furthermore, the code outperforms \ac{DSCF} decoding with the maximum number $T_\text{max}=70$ of bit flips.
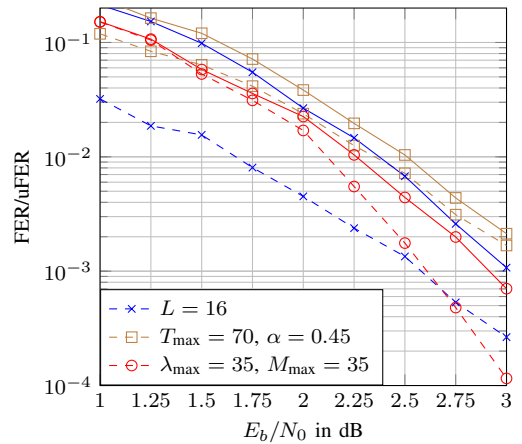
\begin{figure}
	\centering
	\footnotesize
	\begin{tikzpicture}[scale=1]
\begin{semilogyaxis}[
legend columns=1,
legend pos=south west,
ymin=0.0001,
ymax=0.2,
width=3.5in,
height=2in,
grid=both,
legend cell align=left,
xmin = 1.0,
xmax = 3.0,
xlabel = $E_b/N_0$ in dB,
ylabel = FER/uFER,
xtick=data
]

\addplot[blue, dashed, mark = x, mark options={solid}]
table[x=snr,y=fer]{snr fer
1 0.032103
1.25 0.018657
1.5 0.01555
1.75 0.0080541
2 0.0045078
2.25 0.002378
2.5 0.0013448
2.75 0.00053338
3 0.00026507
};\addlegendentry{$L=16$}

\addplot[brown, dashed, mark = square, mark options={solid}]
table[x=snr,y=fer]{snr fer
1 0.119
1.25 0.08375
1.5 0.063667
1.75 0.0415
2 0.023857
2.25 0.012538
2.5 0.00718
2.75 0.0031207
3 0.001678
};\addlegendentry{$T_\text{max}=70$, $\alpha=0.45$}

\addplot[red, dashed, mark = o, mark options={solid}]
table[x=snr,y=fer]{snr fer
1 0.151
1.25 0.105
1.5 0.053
1.75 0.031167
2 0.017
2.25 0.0055
2.5 0.0017609
2.75 0.00048039
3 0.00011538
};\addlegendentry{$\lambda_\text{max}\eddd{\,=\eta}=35$, $M_\text{max}=35$}

\addplot[red, mark = o]
table[x=snr,y=fer]{snr fer
1 0.152
1.25 0.107
1.5 0.058
1.75 0.035833
2 0.0225
2.25 0.0104
2.5 0.004413
2.75 0.0019902
3 0.0007028
};

\addplot[blue, mark = x]
table[x=snr,y=fer]{snr fer
1 0.21284
1.25 0.15261
1.5 0.098429
1.75 0.054929
2 0.026641
2.25 0.014649
2.5 0.0067643
2.75 0.0025922
3 0.0010735
};

\addplot[brown, mark = square]
table[x=snr,y=fer]{snr fer
1 0.2425
1.25 0.16325
1.5 0.12033
1.75 0.071625
2 0.038357
2.25 0.019615
2.5 0.01038
2.75 0.0043707
3 0.0021271
};

\end{semilogyaxis}

\end{tikzpicture}
	\caption{FER \edd{(solid)}/uFER \edd{(dashed)} vs. $E_b/N_0$ over the \ac{biAWGN} channel for \edd{the} $\left(128,64\right)$ \ac{PAC} code under modified \ac{SCOS} decoding with a maximum \acp{PM} and a fixed maximum complexity constraint compared to a $(128,64)$ modified polar code with an outer \ac{CRC}-$7$ having the generator polynomial $g(x)=x^{7}+x^{6}+x^5+x^2+1$\cite{Yuan19}.}
	\label{fig:uFER}
\end{figure}
\section{Conclusions}
\label{sec:conclusions}

The \ac{SCOS} algorithm was proposed that implements \ac{ML} decoding. The complexity adapts to the channel quality and approaches the complexity of SC decoding for the illustrated short- to moderate-length $\mathbf{G}_N$-coset codes with \ac{RM} rate profiles at high \ac{SNR} values. Unlike existing alternatives, the algorithm does not need an outer code or a separate parameter optimization. In addition, it provides better performance compared to SC-Fano and SCS decoding algorithms with a lower complexity at high SNR values.

A modification to \ac{SCOS} was proposed, which provides further gains compared to the original algorithm when there is stringent maximum complexity constraints. The modification has a potential to provide a trade-off between the overall and undetected error probabilities as a byproduct. Using the modified \ac{SCOS}, the $(128,64)$ \ac{PAC} code provides simultaneous gains in the overall and undetected frame error rates compared to \ac{CRC}-concatenated polar codes under \ac{SCL} decoding.

\appendix
\edd{Here, we provide the standard routines recursivelyCalcL and recursivelyCalcC required for SCOS decoding as Algorithms \ref{alg:recursivelyCalcL} and \ref{alg:recursivelyCalcC} similar to \cite[Alg. 3]{tal2015list} and \cite[Alg. 4]{tal2015list}, respectively. Note that Algorithm \ref{alg:recursivelyCalcC} denotes the check and variable node operations as $f^-$ (\texttt{line 8}) and $f^+$ (\texttt{line 10}), where the former may be implemented using min-sum approximation.}
\begin{algorithm}[t]
        \footnotesize
	\SetNoFillComment
	\DontPrintSemicolon
	\SetKwInOut{Input}{Input}\SetKwInOut{Output}{Output}
	\Input{layer $\lambda$ and phase $\phi$}
	\BlankLine
	\lIf{$\lambda=1$}{
		\Return{}
	}
	$\psi = \left \lfloor{\phi/2}\right \rfloor, t = 2^{\lambda-2}$\\
	\If{$\phi \mod 2 = 0$}{
		$\text{recursivelyCalcL}\left(\lambda-1,\psi\right)$
	}
	\For{$\beta=0,1,\dots,2^{\log_2N-\lambda+1}-1$}{
		\eIf{$\phi \mod 2 = 0$}{
			$\texttt{L}\left[\lambda,\phi+2\beta t+1\right]=$\\\nonl~$f^-\left( \texttt{L}\left[\lambda-1,\psi+2\beta t+1 \right],\texttt{L}\left[\lambda-1,\psi+ (2\beta+1) t+1 \right] \right)$
		}{
			$\texttt{L}\left[\lambda,\phi+2\beta t+1\right]=$\\\nonl~$f^+( \texttt{L}\left[\lambda-1,\psi+2\beta t+1 \right],\texttt{L}\left[\lambda-1,\psi+ (2\beta+1) t+1 \right],$\quad$\texttt{C}\left[\lambda,\phi+2\beta t\right])$
		}
	}
\caption{$\text{recursivelyCalcL}\left(\lambda,\phi\right)$}
	\label{alg:recursivelyCalcL}
\end{algorithm}
\begin{algorithm}[t]
        \footnotesize
	\SetNoFillComment
	\DontPrintSemicolon
	\SetKwInOut{Input}{Input}\SetKwInOut{Output}{Output}
	\Input{layer $\lambda$ and phase $\phi$}
	\BlankLine
	$\psi = \left \lfloor{\phi/2}\right \rfloor, t = 2^{\lambda-2}$\\
	\For{$\beta=0,1,\dots,2^{\log_2N-\lambda+1}-1$}{
		$\texttt{C}\left[\lambda-1,\psi+2\beta t + 1\right] = \texttt{C}\left[\lambda, \phi+2\beta t\right]\oplus \texttt{C}\left[\lambda, \phi+2\beta t+1\right]$\\
		$\texttt{C}\left[\lambda-1,\psi+(2\beta+1) t + 1\right] = \texttt{C}\left[\lambda, \phi+2\beta t+1\right]$
	}
	\If{$\psi \mod 2 = 1$}{
		$\text{recursivelyCalcC}\left(\lambda-1,\psi\right)$
	}
	\caption{\edd{$\text{recursivelyCalcC}\left(\lambda,\phi\right)$}}
	\label{alg:recursivelyCalcC}
\end{algorithm}

\section*{Acknowledgements}
The authors would like to thank Gerhard Kramer (TUM) for the comments on an early version of this manuscript. They also thank the Associate Editor and the anonymous reviewers for their valuable comments, which improved the presentation and the content of the work significantly. In particular, the authors thank the reviewer who directed us to use heap structure in the implementations to reduce the complexity.


\end{document}